\documentclass[a4paper,USenglish,cleveref, autoref, thm-restate,authorcolumns,pagebackref]{lipics-v2019}
 \nolinenumbers
 \hideLIPIcs

\usepackage{url}
\usepackage[utf8]{inputenc}
\usepackage{graphicx}
\usepackage{booktabs}
\usepackage{todonotes}
\usepackage{amsmath,amsfonts,amssymb,amsthm} 
\usepackage{mathtools}
\usepackage[autostyle]{csquotes}
\usepackage[numbers,sort,sectionbib]{natbib}
\usepackage[margin=0pt]{subcaption} %
\usepackage{enumerate} %

\usepackage{xcolor,graphicx} %
\usepackage{sidecap}
\sidecaptionvpos{figure}{t}

\usepackage{tabularx,booktabs,multirow} %

\usepackage{multicol}

\hypersetup{pagebackref,menucolor=orange!40!black,filecolor=magenta!40!black,urlcolor=blue!40!black,linkcolor=red!40!black,citecolor=green!40!black,colorlinks}

\usepackage{algorithm}
\usepackage[noend]{algpseudocode}
\algnewcommand\algorithmicinput{\textbf{Input:}}
\algnewcommand\algorithmicoutput{\textbf{Output:}}
\algnewcommand\Input{\item[\algorithmicinput]}
\algnewcommand\Output{\item[\algorithmicoutput]}

\usepackage{tikz}
\usetikzlibrary{arrows,decorations.pathmorphing,decorations.pathreplacing,backgrounds,positioning,fit,matrix}
\usetikzlibrary{shapes,calc}
\usetikzlibrary{graphs,graphs.standard}

\usepackage{pgfplots}
\usepackage{siunitx}
\usepackage{pgfplotstable}
\usepgfplotslibrary{groupplots}
\pgfmathsetseed{420}

\pgfplotsset{compat=1.14}

\pgfmathdeclarefunction{lg2}{1}{%
    \pgfmathparse{ln(#1)/ln(2)}%
}
\pgfmathdeclarefunction{lg10}{1}{%
    \pgfmathparse{ln(#1)/ln(10)}%
}

\tikzstyle{alter}=[circle, minimum size=16pt, draw, inner sep=1pt,fill=white,semithick] 
\tikzstyle{alterlarge}=[circle, minimum size=20pt, draw, inner sep=1pt,fill=white,semithick] 
\tikzstyle{alterinvis}=[circle, minimum size=16pt, draw, inner sep=1pt,draw=white,semithick] 
\tikzstyle{majarr}=[draw=black,semithick]
\tikzstyle{majarrinvis}=[draw=white,semithick]

\usepgflibrary{fpu}  %
\usepgflibrary[fpu]  %
\usetikzlibrary{fpu} %
\usetikzlibrary[fpu] %

\pgfplotsset{
    discard if/.style 2 args={
        x filter/.append code={
            \edef\tempa{\thisrow{#1}}
            \edef\tempb{#2}
            \ifx\tempa\tempb
                
            \else
            \fi
        }
    }
}

\newtheorem{observation}{Observation}
\newtheorem{rrule}{Reduction Rule}
\newtheorem{brule}{Branching Rule}
\newtheorem{lbound}{Lower Bound}

\crefname{rrule}{Reduction Rule}{Reduction Rules} %
\Crefname{rrule}{RR}{RRs} %

\crefname{brule}{Branching Rule}{Branching Rules} 
\Crefname{brule}{BR}{BRs} 

\crefname{lbound}{Lower Bound}{Lower Bounds}
\Crefname{lbound}{LB}{LBs}

\newcommand{\twoclubvertexdelete}{\textsc{2-Club Cluster Vertex Deletion}}

\newcommand{\cstrtwoclubvertexdeletelong}{\textsc{Generalized 2-Club Cluster Vertex Deletion}}
\newcommand{\cstrtwoclubvertexdelete}{\textsc{Gen2CVD}}

\newcommand{\twoclubedit}{\textsc{2-Club Cluster Editing}}

\newcommand{\clustervertexdelete}{\textsc{Cluster Vertex Deletion}}
\newcommand{\weightedtwoclubvertexdelete}{\textsc{Generalized 2-Club Cluster Vertex Deletion}}
\newcommand{\weightedtwoclubvertexdeleteshort}{\textsc{Gen2CVD}}
\newcommand{\dominatingset}{\textsc{Dominating Set}}

\newcommand{\bigO}{\mathcal{O}} %

\newcommand{\NN}{\mathbb{N}} %

\newcommand{\R}{\mathcal{R}}
\newcommand{\I}{\mathcal{I}}

\DeclareMathOperator{\dist}{dist}
\DeclareMathOperator{\robust}{robust}

\newcommand{\poly}{\ensuremath{\operatorname{poly}}}

\newcommand{\NP}{\ensuremath{\operatorname{NP}}}
\newcommand{\coNP}{\ensuremath{\operatorname{coNP}}}

\usepackage{etoolbox}

\newcommand{\appsymb}{$\star$}
\newcommand{\appref}[1]{{\hyperref[proof:#1]{\appsymb}}}
\newcommand{\problemdef}[3]{
  \begin{center}
    \begin{minipage}{0.95\columnwidth}
      \noindent
      \textsc{#1}
      
      \vspace{2pt}
      \setlength{\tabcolsep}{3pt}
      \begin{tabularx}{\columnwidth}{@{}lX@{}}
        \normalsize\textbf{Input:} 		& \normalsize #2 \\
        \normalsize\textbf{Question:} 		& \normalsize #3
      \end{tabularx}
    \end{minipage}
  \end{center}
}

\title{On 2-Clubs in Graph-Based Data Clustering: Theory and Algorithm Engineering}
\newcommand{\tuaddress}{Algorithmics and Computational Complexity, Faculty IV, TU Berlin, Germany}
\author{Aleksander Figiel}{\tuaddress}{aleksander.figiel@campus.tu-berlin.de}{}{Partially supported by DFG NI 369/18.}
\author{Anne-Sophie Himmel}{\tuaddress}{anne-sophie.himmel@tu-berlin.de}{https://orcid.org/0000-0001-7905-7904}{Supported by DFG project NI 369/16.}
\author{André~Nichterlein}{\tuaddress}{andre.nichterlein@tu-berlin.de}{https://orcid.org/0000-0001-7451-9401}{}
\author{Rolf Niedermeier}{\tuaddress}{rolf.niedermeier@tu-berlin.de}{https://orcid.org/0000-0003-1703-1236}{}

\authorrunning{A.~Figiel, A-S.~Himmel, A.~Nichterlein, R.~Niedermeier}

\Copyright{Aleksander Figiel, Anne-Sophie Himmel, André Nichterlein, Rolf Niedermeier}

\ccsdesc[300]{Theory of computation~Graph algorithms analysis}

\keywords{Graph modification problems, parameterized complexity, fixed-parameter tractability, problem kernel, data reduction, branch\&bound, algorithm engineering} 

\EventEditors{John Q. Open and Joan R. Access}
\EventNoEds{2}
\EventLongTitle{42nd Conference on Very Important Topics (CVIT 2016)}
\EventShortTitle{CVIT 2016}
\EventAcronym{CVIT}
\EventYear{2016}
\EventDate{December 24--27, 2016}
\EventLocation{Little Whinging, United Kingdom}
\EventLogo{}
\SeriesVolume{42}
\ArticleNo{23}

\begin{document}

\maketitle

\begin{abstract}
Editing a graph into a disjoint union of clusters is a standard optimization task in graph-based data clustering.
Here, complementing classic work where the clusters shall be cliques, we focus on clusters that shall be 2-clubs, that is, subgraphs of diameter two. This naturally leads to the two NP-hard problems \textsc{2-Club Cluster Editing} (the allowed editing operations are edge insertion and edge deletion) and \textsc{2-Club Cluster Vertex Deletion} (the allowed editing operations are vertex deletions).

Answering an open question from the literature, we show that \textsc{2-Club Cluster Editing} is W[2]-hard with respect to the number of edge modifications, thus contrasting the fixed-parameter tractability result for the classic \textsc{Cluster Editing} problem (considering cliques instead of 2-clubs).
Then focusing on \textsc{2-Club Cluster Vertex Deletion}, which is easily seen to be fixed-parameter tractable, we show that under standard complexity-theoretic assumptions it does not have a polynomial-size problem kernel when parameterized by the number of vertex deletions. 
Nevertheless, we develop several effective data reduction and pruning rules, resulting in a competitive solver, clearly outperforming a standard CPLEX solver in most instances of an established biological test data set.

\end{abstract}

\section{Introduction}
Graph-based data clustering is one of the most important application 
domains for graph modification problems~\cite{SST04}. Roughly speaking,
the goal herein is to transform a given graph into (usually) 
disjoint clusters,
thereby performing as few modification operations 
(edge deletions, edge insertions, 
vertex deletions) as possible. 
This type of problems typically is NP-hard.
The perhaps most prominent problem 
herein is \textsc{Cluster Editing} (also known 
as \textsc{Correlation Clustering}), where the clusters are requested 
to be cliques and one is allowed to perform both edge insertions and 
edge deletions. There has been a lot of work on 
\textsc{Cluster Editing}, e.g., see the surveys by B\"ocker and Baumbach~\cite{BB13} and by Crespelle et al.~\cite{CDFG20}. However, also 
the variant where one modifies the input graph by vertex deletions 
received significant interest~\cite{Bor+16,DK12,Huf+10,Tsu19}. 

Arguably, for many data science applications the request that 
the clusters have to be cliques is too rigid. 
Hence, the consideration of clique relaxations for defining clusters 
gained attention in graph-based data 
clustering~\cite{BMN12,Guo+10,LZZ12,MPS13}.
In this work, we focus on so-called \emph{2-clubs} as 
clusters~\cite{LZZ12,MPS13}: 
these are diameter-at-most-two graphs (hence, cliques are 1-clubs). 
Other than finding cliques, finding 2-clubs of size at least~$k$ is fixed-parameter
tractable with respect to~$k$~\cite{SKMN12,HKN15a}.
Note that 2-clubs already have been used in the context of
biological data analysis~\cite{JGG+18,Pas08}.
Moreover, 2-clubs have been studied in the context of covering 
vertices in a graph~\cite{DL19,DMSZ19,DMZ19}.

Now, continuing and complementing previous work of Liu et al.~\cite{LZZ12},
we study both the edge editing variant (referred to as \textsc{2-Club Cluster Editing}) and the vertex deletion variant (referred to as \textsc{2-Club Cluster Vertex Deletion}). 
We contribute
the following three main results:
\begin{enumerate}
\item Answering an open question of Liu 
et al.~\cite{LZZ12}, in \cref{sec:2ccedit}
we show that \textsc{2-Club Cluster Editing} is W[2]-hard with 
respect to the number of modified edges (deletions and insertions),
hence most likely not fixed-parameter tractable. 
This stands in sharp contrast to the problems 
\textsc{Cluster Editing}~\cite{GGHN05} and the more general \textsc{$s$-Plex Cluster Editing}~\cite{Guo+10}\footnote{This is the generalization of \textsc{Cluster Editing} 
where clusters are requested to be $s$-plexes (and not cliques); an $s$-plex is
a subgraph where each vertex is connected to all other vertices of the 
$s$-plex except for at most $s-1$ vertices. Notably, a clique is a 1-plex.},
both known to be fixed-parameter tractable for the parameter number of edge modifications.
The W[2]-hardness seems surprising considering the fact that while 
\textsc{Cluster Editing} is fixed-parameter tractable~\cite{BB13} and 
\textsc{2-Club Cluster Editing} is presumably not, by way of contrast 
finding cliques is presumably not fixed-parameter tractable
while finding 2-clubs is. 
\item Complementing fixed-parameter tractability and kernelization
results for \textsc{Cluster Vertex Deletion}~\cite{Bor+16,Huf+10,Tsu19} and
\textsc{$s$-Plex Cluster Vertex Deletion}~\cite{BMN12}, 
in \cref{sec:no-poly-kernel} we show that, other than 
these related problems and despite being easily seen to be 
fixed-parameter tractable for the parameter solution size, 
\textsc{2-Club Cluster Vertex Deletion} is 
unlikely to have a polynomial-size problem kernel.\footnote{It has been featured as an open problem whether the edge deletion variant \textsc{$s$-Club Cluster Edge Deletion}
has a polynomial-size problem kernel~\cite{CDFG20,LZZ12}.}
\item In \cref{sec:algorithms,sec:implementation,sec:impl-exp}, 
we explore the fixed-parameter tractability of \textsc{2-Club Vertex Deletion} from a more practical angle and develop several efficient 
data reduction rules together with effective search-tree pruning rules. 
Performing an empirical evaluation with standard biological data,
we show that our tuned algorithmic approach 
(based on branching and data reduction) in most relevant cases 
clearly outperforms a CPLEX-based solver, thus providing a state-of the art software tool for the vertex deletion variant of graph-based data clustering with 2-clubs.
\end{enumerate}

\subsection{Preliminaries}
All graphs considered in our work are undirected and simple.
For a graph~$G = (V, E)$ we set~$n := |V|$ and~$m:=|E|$.
We denote with~$\binom{V}{2}$ the set of all two-element subsets of~$V$.

For a vertex~$v\in V$, we denote by~$N_G(v):=\{ w \in V \mid \left\{ v, w \right\} \in E \}$ the \emph{open neighborhood} of~$v$ and by~$N_G[v]:=N_G(v) \cup \left\{ v \right\}$ the \emph{closed neighborhood} of~$v$. 
The \emph{degree} of~$v$ is~$\deg_G(v):= |N_G(v)|$. 
For a vertex subset~$V' \subseteq V$, let~$N_G[V'] := \bigcup_{v \in V'} N_G[v]$.
If it is clear from the context, then we omit~$G$ from the subscripts. 
We denote by~$G[V']$ the subgraph of~$G$ induced by the vertex set $V' \subseteq V$ and by $G[E']$ the subgraph of~$G$ induced by the edge set $E' \subseteq E$, that is, $G[E']:=(V, E')$. 
The graph~$G-v$ is obtained by deleting~$v \in V$ from~$G$, that is~$G-v := G[V \setminus \{v\}]$.

A path~$P$ in~$G$ is an ordered sequence of pairwise distinct vertices~$v_1, v_2, \dots, v_{k+1} \in V$ such that~$\{ v_i, v_{i+1} \} \in E$ for all~$i \in \{1, \dots, k \}$. 
It is also an \emph{induced path} if these are the only edges between its vertices. 
The length of~$P$ is~$k$. 
We will call a path on~$n$ vertices a~$P_n$.
The \emph{distance} of two vertices~$s, t \in V$, denoted by~$\dist_G(s, t)$, is the length of a shortest path connecting~$s$ and~$t$ if one exists, and~$\infty$ otherwise.
The \emph{diameter} of a graph is the maximum distance of any two vertices, formally~$\max_{s, t \in V} \dist_G(s, t)$.
A graph is said to be \emph{connected} if there exists a path between all pairs of its vertices.
A (connected) \emph{component} of a graph~$G$ is a maximal vertex set~$S \subseteq V$ such that~$G[S]$ is connected. 

\subparagraph*{$s$-Club.}
An~\emph{$s$-club} is a graph of diameter at most~$s$. 
A \emph{clique} is a $1$-club.
Furthermore, an \emph{$s$-club cluster graph} is a graph in which each component is an~$s$-club.
In this paper, we consider the following two problems, where $E \triangle F :=(E \setminus F) \cup (F \setminus E)$ denotes the \emph{symmetric difference} of two sets.
\problemdef{\textsc{$s$-Club Cluster Editing}}
	{An undirected graph~$G=(V,E)$ and an integer~$k \in \NN$.}
	{Is there an edge set~$F \subseteq \binom{V}{2}$ with~$|F| \leq k$ such that~$G[E \triangle F]$ is an~$s$-club cluster graph?}
\problemdef{\textsc{$s$-Club Cluster Vertex Deletion}}
	{An undirected graph~$G=(V,E)$ and an integer~$k \in \NN$.}
	{Is there a vertex subset~$S \subseteq V$ with~$|S| \leq k$ such that~$G[V \setminus S]$ is an~$s$-club cluster graph?}
An edge set~$F\subseteq \binom{V}{2}$ such that $G[E \triangle F]$ is an $s$-club cluster graph is called an \emph{$s$-club editing set} and a vertex set~$S\subseteq V$ such that $G[V \setminus S]$ is an $s$-club cluster graph is called an \emph{$s$-club vertex deletion set}.

\subparagraph*{2-Club.}\label{sect:2clubspecifics}
A 2-club is a graph with diameter at most two.
This means that for all pairs of vertices~$u, v \in V$ it holds that~$u$ and~$v$ are adjacent or have at least one common neighbor. 
Note that 2-clubs are \emph{non-hereditary}, that is, if~$G$ is a 2-club, then deleting vertices from~$G$ may destroy this property.
This is a significant difference in comparison with cliques.

Using terminology of \citet{LZZ12}, we call a path~$stuv$ in~$G$ a \emph{restricted~$P_4$} if~$\dist_G(s, v) = 3$. 
That is, a restricted~$P_4$ is a shortest path connecting~$s$ and~$v$ and is thus also an induced~$P_4$. 
The following characterization is easy to verify:
\begin{observation}[{\cite[Lemma 3]{LZZ12}}]
	\label[observation]{obs:charact-2cc-graph}
	A graph~$G$ is a 2-club cluster graph if and only if it contains no restricted~$P_4$.
\end{observation}

\subsection{Parameterized Algorithmics}
A parameterized problem~$\Pi \subseteq \Sigma^* \times \Sigma^*$ is a set of pairs~$(I, k)$, where~$I$ denotes the problem instance and~$k$ is the parameter. 
Problem~$\Pi$ is \emph{fixed-parameter tractable} (FPT) if there exists an algorithm solving any instance of~$\Pi$ in~$f(k) \cdot |I|^c$ time, where~$f$ is some computable function and~$c$ is some constant. 
A \emph{parameterized reduction} from a parameterized problem~$\Pi \subseteq \Sigma^* \times \Sigma^*$ to a parameterized problem~$\Pi' \subseteq \Sigma^* \times \Sigma^*$ is a function which maps any instance~$(I,k) \in \Sigma^* \times \Sigma^*$ to another instance~$(I',k')\in \Sigma^* \times \Sigma^*$ such that
(1) $(I',k')$ can be computed from~$(I,k)$ in FPT time.%
(2) $k' \leq g(k)$ for some computable function~$g$, and 
(3) $(I,k) \in \Pi \iff (I',k') \in \Pi'$. 
If~$\Pi$ is W[$i$]-hard, $i \ge 1$, then such a parameterized reduction shows that also~$\Pi'$ is W[$i$]-hard, that is, presumably not fixed-parameter tractable.
A \emph{reduction to a problem kernel} is a parameterized self-reduction (from~$\Pi$ to~$\Pi$) such that~$(I', k')$ can be computed in polynomial time and~$|I'| \le g(k)$.
If~$g$ is a polynomial, then~$(I', k')$ is called a \emph{polynomial kernel}.
Problem kernels are usually achieved by applying \emph{data reduction rules}.
Given an instance~$(I, k)$, a data reduction rule computes in polynomial time a new instance~$(I', k')$.
We call a data reduction rule \emph{safe} if $(I, k) \in \Pi \iff (I', k') \in \Pi$.

\subsection{Organization of the paper}
We prove in \cref{sec:2ccedit} the W[2]-hardness of \twoclubedit{}.
In \cref{sec:no-poly-kernel}, we show that \twoclubvertexdelete{} does not admit a polynomial kernel with respect to~$k$ unless $\NP\subseteq \coNP/\poly$.
We then present in \cref{sec:algorithms} an ILP-formulation and a branch\&bound algorithm for \twoclubvertexdelete{}.
In \cref{sec:implementation}, we provide implementation details for the solver which we experimentally evaluate in \cref{sec:impl-exp}.
We conclude in \cref{sec:conclusion}.

\section{W[2]-Hardness of 2-Club Cluster Editing}
\label{sec:2ccedit}
It is easy to see that \twoclubvertexdelete{} is fixed-parameter tractable with respect to solution size~$k$ \cite{LZZ12}:
By \cref{obs:charact-2cc-graph}, it is enough to recursively search for a restricted~$P_4$~$stuv$ and delete a vertex to separate~$s$ and~$v$. 
In contrast, we subsequently show that \twoclubedit{} is W[2]-hard with respect to solution size~$k$ answering an open question of \citet{LZZ12}.
Intuitively, the hardness is due to the fact that there is a ``non-local'' way of destroying a restricted~$P_4$ with edge insertions, see \cref{fig:p4_resolutions_example} for an illustration.

\newcommand{\tworows}[2]{\begin{tabular}{l}{#1}\\{#2}\end{tabular}}

\begin{figure}[t]
	\centering
	\begin{tikzpicture}
		\node at (-2,0.5) {\tworows{six local}{modifications:}};
	
		\node[alter] at (0,0) (s) {$s$};
		\node[alter, right = 4ex of s] (t) {$t$};
		\node[alter, right = 4ex of t] (u) {$u$};
		\node[alter, right = 4ex of u] (v) {$v$};

		\draw[majarr, densely dotted] (s) edge (t);
		\draw[majarr, densely dotted] (t) edge (u);
		\draw[majarr, densely dotted] (u) edge (v);

		\draw[majarr, dashed] (s) to[out=45,in=135,looseness=1.0] (u);
		\draw[majarr, dashed] (v) to[out=135,in=45,looseness=1.0] (t);
		\draw[majarr, dashed] (s) to[out=50,in=130,looseness=1.3] (v);
	
		\begin{scope}[yshift = -2cm]
			\node at (-2,0.5) {\tworows{non-local}{modification:}};
			\node[alter] at (0,0) (s) {$s$};
			\node[alter, right = 4ex of s] (t) {$t$};
			\node[alter, right = 4ex of t] (u) {$u$};
			\node[alter, right = 4ex of u] (v) {$v$};

			\draw[majarr] (s) edge (t);
			\draw[majarr] (t) edge (u);
			\draw[majarr] (u) edge (v);

			\coordinate (Middle) at ($(s)!0.5!(v)$);

			\node[alter, above = 4ex of Middle] (b) {$b$};
			\draw[majarr, dashed] (b) edge (s);
			\draw[majarr, dashed] (b) edge (v);
		\end{scope}

		\begin{scope}[xshift = 6.5cm, yshift = -0.6cm]
			\def\n{5}
			\foreach \x [count = \i] in {a,...,e} {
					\node[alter] (A-\i) at ({360 * (\i-1) / \n + 90}:12ex) {$\x$};
			}
			\foreach \x [count = \i] in {f,...,j} {
					\node[alter] (B-\i) at ({360 * (\i-1) / \n + 90}:6ex) {$\x$};
			}

			\foreach \i [evaluate={\j=int(mod(\i,5)+1)}] in {1,2,4,5} {
				\draw[majarr] (A-\i) edge (A-\j);
			}
			\draw[majarr] (A-3) edge (A-4) [dashed];

			\foreach \i [evaluate={\j=int(mod(\i+2+4,5)+1)}]
			in {1,2,3,4,5}{
				\draw[majarr] (A-\i) edge (B-\i);
				\draw[majarr] (B-\j) edge (B-\i);
			}
		\end{scope}
	\end{tikzpicture}
	\caption{
		\emph{Left (top and bottom):} All possible modifications to destroy a restricted~$P_4$. 
		Top: The six ``local'' modifications; that is, any edge which is inserted (dashed edges) or deleted (dotted edges) has both its ends in the~$P_4$. %
		Bottom: A ``non-local'' modification (the two inserted edges are dashed), where~$b$ can be any vertex other than~$s, t, u$ and~$v$.
		\emph{Right side:}
		The dashed edge indicates the single optimal solution (inserting the edge, the resulting Petersen graph has diameter two) which is a non-local modification. 
		Note that the distance of~$c$ and~$d$ before was four. 
		Hence, inserting the edge~$\left\{ c,d \right\}$ is not part of any local modification. 
	}
	\label{fig:p4_resolutions_example}
\end{figure}
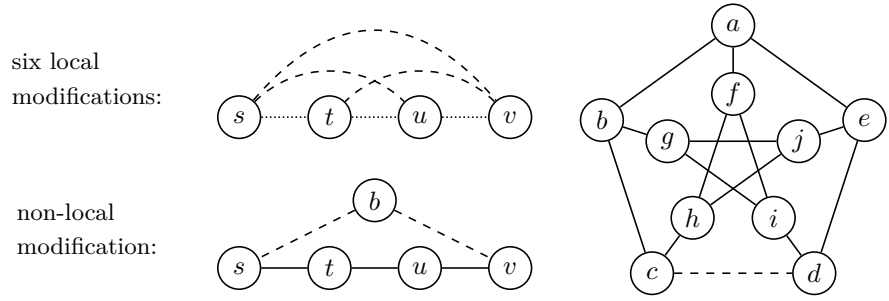

The basic idea of our parameterized reduction from \dominatingset{}\footnote{Given an undirected graph~$G = (V,E)$ and an integer~$k$, the question is whether there is a dominating set~$V' \subseteq V$ (that is, $N[V'] = V$) of size at most~$k$.} is inspired by a parameterized reduction by~\citet[Theorem 1]{GHN13} who showed hardness for the problem of reducing the diameter of a given graph to two by inserting at most~$k$ edges.
In our reduction we need to take care of the possibility to delete edges, which changes many details of the construction.
\dominatingset{} remains W[2]-hard with respect to~$k$ for graphs with diameter two \cite{Lok+13}, which allows us to assume that the \dominatingset{} instance has diameter two.

\begin{theorem} \label{thm:2clubw2hard}
\twoclubedit{} is W[2]-hard with respect to~$k$.
\end{theorem}

\begin{proof}
Let~$G = (V, E)$ be a graph with diameter two. 
We construct a graph~$G'$ in such a way that~$G$ has a dominating set of size at most~$k$ if and only if~$G'$ has a 2-club editing set of size at most~$k$.
The graph~$G' = (V', E')$ can be broken down into the following parts: the original graph~$G$, a clique~$C \subseteq V'$ of cardinality~$(n+1)^2$, and a single vertex~$x$.
We assign two indices for the vertices~$c_{i, j} \in C$ such that~$i, j \in \left\{ 0, \dots, n \right\}$. 
The vertices in~$V$ only have one index:~$v_i \in V$, $i \in \left\{ 1, \dots, n \right\}$.
In addition to the existing edges in~$G$ and~$C$, add the following edges:
for each~$j \in \left\{ 0, \dots, n \right\}$ add~$\left\{ x, c_{0,j} \right\}$ and
for each~$c_{i,j} \in C, i \neq 0$, add~$\left\{ v_i, c_{i,j} \right\}$.
The graph~$G'$ has~$\bigO(n^4)$ edges and~$\bigO(n^2)$ vertices.
For a schematic picture of~$G'$ see \cref{fig:2clubw2hardexample}. 
Note that the only pairs of vertices with distance three are~$x$ and~$v_i \in V$, all others have distance at most two.
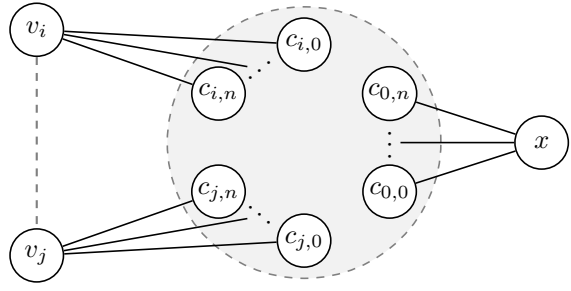
\begin{figure}[t!]
	\centering
	\begin{tikzpicture}[auto]
		\foreach \x / \name [count=\i] in {ci0/$c_{i,0}$, cin/$c_{i,n}$, cjn/$c_{j,n}$, cj0/$c_{j,0}$, c00/$c_{0,0}$, c0n/$c_{0,n}$} {
				\node[alterlarge] (\x) at ({360 * (\i-1) / 6 + 90}:1.3) {\name};
		}
		\node[rotate= 40] (cid) at ($(ci0)!0.5!(cin)$) {\dots};
		\node[rotate=  -40] (cjd) at ($(cj0)!0.5!(cjn)$) {\dots};
		\node[rotate=90] (c0d) at ($(c00)!0.5!(c0n)$) {\dots};

		\path (cid) ++(170:3) node[alterlarge] (vi) {$v_i$};
		\path (c0d) ++(  0:2) node[alterlarge] (x) {$x$};
		\path (cjd) ++(190:3) node[alterlarge] (vj) {$v_j$};
		\draw[gray,thick, dashed] (vi) -- (vj);			
		
		\draw[majarr] (x) edge (c00);
		\draw[majarr] (x) edge (c0d);
		\draw[majarr] (x) edge (c0n);

		\draw[majarr] (vi) edge (ci0);
		\draw[majarr] (vi) edge (cid);
		\draw[majarr] (vi) edge (cin);
		\draw[majarr] (vj) edge (cj0);
		\draw[majarr] (vj) edge (cjd);
		\draw[majarr] (vj) edge (cjn);

		\begin{pgfonlayer}{background}
			\filldraw[fill=gray!10,rounded corners=15mm, draw=black!50, dashed, semithick]
				(0,0) circle (1.8);
		\end{pgfonlayer}

	\end{tikzpicture}
	\caption{
		A schematic picture of the construction of~$G'$ in the proof of \cref{thm:2clubw2hard}. 
		The vertices in the gray circle form a clique, but only the vertices connected to~$v_i, v_j,$ or~$x$ are shown. 
		The dashed gray edge between~$v_i$ and~$v_j$ exists if $\{v_i,v_j\} \in E(G)$.}
	\label{fig:2clubw2hardexample}
\end{figure}

We claim that there exists a 2-club editing set of size at most~$k$ for~$G'$ (which only inserts edges) if and only if there exists a dominating set of size at most~$k$ for~$G$.

``$\Leftarrow$'': Let~$D$ be a dominating set for~$G$ with $|D|\leq k$, and~$F := \{ \{x, v\} \allowbreak  \mid \allowbreak v \in D \}$. 
Let~$H := G'[E' \triangle F ]$. 
For every~$v_i \in V$, either $v_i \in D$ and then~$\dist_H(x, v_i) = 1$, or~$v_i \notin D$ and then~$v_i$ has a neighbor in~$D$ and thus~$\dist_H(x, v_i) = 2$. 
This means that~$H$ is a 2-club cluster graph and~$F$ is a 2-club editing set for~$G'$ with $|F| \leq k$.
 
``$\Rightarrow$'': 
Let~$F$ be a 2-club editing set for~$G'$ with $|F| \leq k$ and~$H = G'[E' \triangle F ]$ be the resulting 2-club cluster graph.
Assume without loss of generality that~$F$ is minimal.
Note that the minimum cut of~$G'$ is~$n+1$ and that $k < n$. 
Removing any edge would only be optimal if~$H$ contained more than one 2-club cluster. 
Hence, we can assume that the 2-club editing set~$F$ does not delete any edges from~$G'$.

For any inserted edge~$\left\{ a, b \right\} \in F$ exactly one of the following cases applies, since the distance between~$x$ and some~$v_i \in V$ has to be reduced by means of inserting~$\left\{ a, b \right\}.$
\begin{itemize}
\item $\left\{ a, b \right\} = \left\{ v_i, x \right\}$:
	Then~$\dist_{H}(x, v_i) = 1$ and for~$a \in N_G(v_i)$~$\dist_{H}(x, a) \leq 2$. We interpret this as~$v_i$ being a dominating vertex in~$G$.

\item $\left\{ a, b \right\} = \left\{ v_i, c_{0,j} \right\}$:
	This edge enables a path of length two from~$v_i$ to~$x$ via~$c_{0,j}$. 
	This means that this edge is only of benefit to~$v_i$. 
	Then~$F' = (F \setminus \left\{ v_i, c_{0,j} \right\}) \cup \left\{ x, v_i \right\}$ is also a 2-club editing set with~$|F| = |F'|$.

\item $\left\{ a, b \right\} = \left\{ v_i, v_j \right\}$: 
	This means that one of the vertices has an edge to~$x$. 
	Without loss of generality assume that~$\left\{x, v_i \right\} \in F$. 
	Note that~$F$ is only minimal if~$\left\{ x, v_j \right\} \notin F$, as the edge~$\left\{ v_i, v_j \right\}$ is only of benefit to~$v_j$ and no other vertices since it enables a path of length two from~$v_j$ to~$x$ via~$v_i$. 
	Then~$F' = (F \setminus \left\{ v_i, v_j \right\}) \cup \left\{ x,  v_j\right\}$ is also a  2-club editing set with~$|F| = |F'|$.

\item $\left\{ a, b \right\} = \left\{ v_i, c_{j,k} \right\}, j \neq i, j \neq 0$:
	This means that there is an edge~$\left\{ x, c_{j,k} \right\} \in F$, otherwise~$F$ would not be minimal. 
	The edge~$\left\{ v_i, c_{j,k} \right\}$ enables a path of length two from~$v_i$ to~$x$ via~$c_{j,k}$. 
	This means that the edge is of no benefit to any other vertices. 
	Then~$F' = F \setminus \left\{ v_i, c_{j,k} \right\} \cup \left\{ x, v_i \right\}$ is also a 2-club editing set with~$|F| = |F'|$.

\item $\left\{ a, b \right\} = \left\{ x, c_{i,j} \right\}, i \neq 0$: 
	This edge enables a path of length two from~$v_i$ to~$x$ via~$c_{i,j}$. 
	In the previous case, we have seen that there exists an~$F'$ with~$\left\{ x, c_{i,j} \right\} \in F'$ such that there exists no edge~$\left\{ v_k, c_{i,j} \right\} \in F'$ with~$k \neq i$. 
	This means that the edge~$\left\{ x, c_{i,j} \right\}$ is of no benefit to any other vertices. 
	Then~$F'' = (F' \setminus \left\{ x, c_{i,j} \right\}) \cup \left\{ x, v_i \right\}$ is also a 2-club editing set with~$|F| = |F''|$.
\end{itemize}

Altogether, we know that there exists an~$F'$ with~$|F'| = |F|$ such that~$F'$ is a 2-club editing set of the form~$\left\{ \left\{ x,v \right\} \mid v \in D \right\}$ for some~$D \subseteq V$. This means that~$D$ is a dominating set for~$G$ with $|D| \leq k$.

Summarizing, the reduction from~$(G, k)$ to~$(G', k)$ is a valid parameterized reduction from \dominatingset{} for graphs with diameter two to \twoclubedit{}. 
Since \dominatingset{} is W[2]-hard for graphs of diameter two~\cite{Lok+13}, this yields that \twoclubedit{} is also W[2]-hard.
 \end{proof}

\section{No Polynomial Kernel for 2-Club Cluster Vertex Deletion} \label{sec:no-poly-kernel}

We use the OR-cross-composition framework of Bodlaender et al.~\cite{BJK14} to show that \twoclubvertexdelete{} admits no polynomial kernel with respect to~$k$.
Given an \NP-hard problem~$L$, an equivalence relation~$\R$ on the instances of~$L$ is a \emph{polynomial equivalence relation} if 
\begin{enumerate}[(i)]
 \item one can decide for any two instances in time polynomial in their sizes whether they belong to the same equivalence class, and
 \item for any finite set~$S$ of instances, $\R$ partitions~$S$ into at most~$(\max_{x \in S} |x|)^{O(1)}$ equivalence classes. %
\end{enumerate}
\begin{definition} 
	Given an \NP-hard problem~$L$, a parameterized problem $P$, and a polynomial equivalence relation~$\R$ on the instances of $L$, an \emph{OR-cross-composition} of $L$ into $P$ (with respect to $\R$) is an algorithm that takes $\ell$ $\R$-equivalent instances $\I_1,\ldots,\I_\ell$ of $L$ and constructs in time polynomial in $\sum_{i=1}^\ell |\I_\ell|$ an instance $(\I,k)$ such that 
	\begin{enumerate}
		\item $k$ is polynomially upper-bounded in $\max_{1\leq i\leq \ell}|\I_i|+\log(\ell)$ and 
		\item $(\I,k)$ is a yes-instance for $P$ if and only if there is at least one~$\ell'\in[\ell]$ such that $\I_{\ell'}$ is yes-instance for $L$. 
	\end{enumerate}
\end{definition}
If a parameterized problem~$P$ admits an OR-cross-composition for some \NP-hard problem~$L$, then~$P$ does not admit a polynomial kernel with respect to its parameterization, unless $\NP\subseteq \coNP/\poly$~\cite{BJK14}.

\begin{theorem}
	\label{thm:nopolykernel}
	\twoclubvertexdelete{} does not admit a polynomial kernel with respect to~$k$ unless $\NP\subseteq \coNP/\poly$.
\end{theorem}

\begin{proof}
	To show the result, we provide an OR-cross-composition from \twoclubvertexdelete{} to itself.
	To this end, we define~$\R$ as follows: two instances~$(G_1, k_1)$, $(G_2, k_2)$ of \twoclubvertexdelete{} are equivalent with respect to $\R$ iff~$k_1 = k_2$.
	Since the solution size is at most~$n$, this gives a polynomial equivalence relation.
	
	Given~$\ell$ $\R$-equivalent instances~$(G_1 = (V_1, E_1), k), \ldots, (G_\ell = (V_\ell, E_\ell),k)$, we construct a new instance~$(G' = (V', E'), k')$ as follows.
	Without loss of generality, assume that~$\ell$ is a power of two (otherwise copy instances until~$\ell$ is a power of two).
	We set~$k' := k + \log \ell$.
	To describe~$G'$, we need a simple \emph{selection-gadget} consisting of two stars with~$k'+1$ leaves each where the two center vertices are adjacent.
	Observe that in the selection-gadget the leaves of one star are at distance three to the leaves of the other star.
	Moreover, since each star has more than~$k'$ leaves, the only possibility to transform a selection-gadget into a 2-club cluster graph is to delete one of the two center vertices.
\begin{figure}[t!]
	\tikzstyle{simple}=[inner sep=1pt,semithick] 
	\tikzstyle{alter}=[star,star points=8,minimum size =10pt, star point ratio=2, draw, inner sep=1pt,semithick,fill=white]
	{
	\begin{center}
		\includegraphics{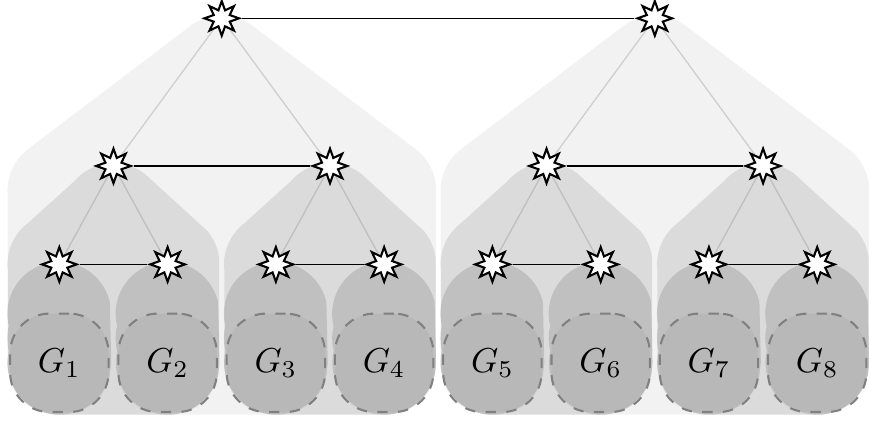}
	\end{center}
	}
	\caption{Illustration of the construction for \cref{thm:nopolykernel} exemplified for~$\ell=8$. Star-shaped vertices have~$k'+1$ additional leaves and are connected to all vertices in the gray-shaded area below them.}
	\label{figure:nopolykernel}
\end{figure}

	We can now define~$G'$: 
	To this end, we recursively create an ``instance-selector'' that forces the selection of exactly one instance~$G_i$ as shown in \cref{figure:nopolykernel}.
	First, add a selection-gadget with the two center vertices~$c_L$ and~$c_R$ (left and right).
	Second, recursively build the two graphs~$G_L, G_R$ composing~$G_1, \ldots, G_{\ell/2}$ and~$G_{\ell/2 + 1}, \ldots, G_\ell$ respectively until $G_L,G_R$ consist of only one input instance.
	Make every vertex in~$G_L$ (in~$G_R$) adjacent to~$c_L$ (to~$c_R$).
	Note that this recursive procedure has recursion depth~$\log \ell$.
	
	The construction of~$(G',k')$ can clearly be done in polynomial time.
	It remains to show the correctness, that is, $(G', k')$ is a yes-instance if and only if there is a yes-instance~$(G_i,k)$, $i \in [\ell]$.
	
	``$\Rightarrow$:''
	Let~$S' \subseteq V'$ be a minimal solution of size at most~$k + \log \ell$ for~$G'$.
	Observe that by construction of~$G'$ at least one of the two center vertices~$c_L$ and~$c_R$ of the ``topmost'' selection-gadget has to be in~$S'$.
	Assume without loss of generality that~$c_L \in S$ (the other case is completely analogous).
	Observe that the connected component~$C_R$ in~$G' - c_L$ that contains~$c_R$ is a 2-club since~$c_R$ is a universal vertex in~$C_R$.
	Since~$S'$ is minimal, it follows that~$S'$ contains no vertex in~$C_R$.
	By construction, the connected component containing the graphs~$G_1, \ldots, G_{\ell / 2}$ contains a selection-gadget where again one of the two center vertices has to be in~$S'$.
	By induction on the recursion depth one can show that~$S'$ contains exactly~$\log \ell$ center vertices of the selection-gadgets.
	Moreover, there is exactly one graph~$G_i$ such that all~$\log \ell$ center vertices adjacent to the vertices in~$G_i$ are in~$S'$.
	Since~$S'$ is a solution for~$G'$ it follows that~$S' \cap V_i$ is a solution of size at most~$k$ for~$G_i$.

	``$\Leftarrow$:''
	Let~$i \in [\ell]$ be such that~$(G_i,k)$ is a yes-instance and let~$S \subseteq V_i$ be the solution for~$G_i$, $|S| \le k$. 
	The solution~$S' \subseteq V'$ for~$G'$ consists of~$S$ and every adjacent center vertex in a selection-gadget.
	Observe that~$|S'| \le k + \log \ell$ since, by construction, there are~$\log \ell$ selection-gadgets that contain a vertex adjacent to~$G_i$.
	Thus, it remains to show that $G' - S'$ is a 2-club cluster graph.
	To this end, observe that, by assumption, $G_i - S$ is a 2-club cluster graph. 
	Note that each graph~$G_j$, $i \ne j$, is in~$G' - S'$ in a connected component with a center vertex~$c$ of a selection-gadget such that the other center vertex of this gadget is in~$S'$.
	Observe that each such center vertex~$c$ is a universal vertex in its connected component in~$G' - S'$.
	Thus this connected component containing~$c$ forms a 2-club.
	Since each connected component of~$G' - S'$ contains vertices of at least one graph~$G_j$, it follows that~$G' - S'$ is a 2-club cluster graph.
 \end{proof}

\section{Algorithms for 2-Club Cluster Vertex Deletion} \label{sec:algorithms}

In this section, we first formulate \twoclubvertexdelete{} as an Integer Linear Program (ILP) and then introduce a branch\&bound-algorithm solving a generalization of \twoclubvertexdelete{}.
We use the ILP-formulation in our experiments to evaluate our branch\&bound algorithm.

\subsection{ILP Formulation}\label{sec:ILP}
By \cref{obs:charact-2cc-graph}, a graph is a 2-club cluster graph if and only if it contains no restricted~$P_4$. 
Recall that a restricted~$P_4$ is an induced~$P_4$ $stuv$ that is also a shortest path between~$s$ and~$t$.
Thus, there exists no vertex~$w \in N(s)\cap N(v)$ in the common neighborhood of~$s$ and~$v$.
The deletion of a vertex cannot create any new induced path but it might \enquote{promote} an induced~$P_4$ to a restricted~$P_4$.
Hence, if~$N(s) \cap N(v) = \emptyset$ for any induced~$P_4$~$stuv$ in~$G$, then at least one vertex from~$stuv$ must be deleted.

We introduce a variable~$x_v$ for each vertex~$v \in V$. This variable has a value of~$1$ if and only if~$v$ is in the 2-club vertex deletion set.
This leads to the following ILP formulation:
\begin{align*}
& \text{min:}	& & \sum_{v \in V} x_v \\
& \text{s.t.}	& & x_s + x_t + x_u + x_v + \smashoperator{\sum_{b \in N(s) \cap N(v)}} (1-x_b) \geq 1 	& & \text{ for all induced } P_4 \text{'s } stuv \text{ in } G \\
&      		& & x_v \in \left\{ 0, 1 \right\} 									& & \text{ for all } v \in V.
\end{align*}

\subsection{Branch\&Bound Algorithm} \label{app:sec:BB}

In this section we provide details on our branch\&bound algorithm for a slightly more general variant of \twoclubvertexdelete{} that allows more flexibility in the design of data reduction rules and for deriving lower bounds.
Based on a simple search-tree, this algorithm extensively uses data reduction rules and lower bounds. 
While these lower bounds and data reduction rules give a significant speedup in practice as shown in \cref{sec:impl-exp} (also cf.\, \citet{KNN15}), we could not show an improved theoretical worst-case bound.

\subsubsection{Search Tree} \label{ssec:search-tree}

A graph is a 2-club cluster graph if and only if there exists no restricted~$P_4$ (\cref{obs:charact-2cc-graph}).
This observation yields a straight-forward~$O^*(4^k)$ search tree algorithm for \twoclubvertexdelete.
To shrink the search tree, we introduce the concept of \emph{permanent vertices}. A vertex is permanent if it is not allowed to be removed from the graph.  
Additionally, we introduce weights on the vertices.
These two concepts will allow us to apply a wider range of data reduction rules and lower bounds in the search tree algorithm.
The resulting problem is defined as follows: 

\problemdef{\cstrtwoclubvertexdeletelong{} (\cstrtwoclubvertexdelete{})}
	{An undirected graph~$G = (V, E)$, an integer~$k \in \NN$, a set~$F \subseteq V$ of permanent vertices, and a weight function~$w\colon V \rightarrow \NN^+$.} %
	{Is there an~$S \subseteq V$ with~$w(S) \leq k$ and~$S \cap F = \emptyset$ such that~$G[V \setminus S]$ is a 2-club cluster graph?}

	Note that an instance~$(G, k)$ of \twoclubvertexdelete{} is clearly equivalent to the instance~$(G, k, \emptyset, w \equiv 1)$ of \cstrtwoclubvertexdelete{}. 

Our algorithm uses a simple branching rule that takes a restricted~$P_4$ and branches into all four cases of deleting one vertex which implies updates of the set~$F$ of permanent vertices in each branch.
If some vertex of the restricted~$P_4$ $stuv$ is already in~$F$, then we skip the corresponding case in the branching.
Thus, the branching itself ``grows'' the set~$F$ of permanent vertices that will reduce the cases to be considered later in the branching.
Moreover, if more than one restricted~$P_4$ exists, then the algorithm chooses one with most vertices in~$F$ and uses the weights of the vertices as tiebreaker.
\begin{brule} \label{br:2cvd_overlap}
	Let~$\mathcal{I} = (G, k, F, w)$ be an instance of \cstrtwoclubvertexdelete{}. If~$G$ is not a 2-club cluster graph, then find a restricted~$P_4$~$stuv$ and split~$\mathcal{I}$ into four smaller instances~$\mathcal{I}_s, \mathcal{I}_t, \mathcal{I}_u, \mathcal{I}_v$ as follows:
	\begin{itemize}
	\item $\mathcal{I}_s = (G-s, k-w(s), F, w)$,
	\item $\mathcal{I}_t = (G-t, k-w(t), F \cup \left\{ s \right\},w)$,
	\item $\mathcal{I}_u = (G-u, k-w(u), F \cup \left\{ s, t \right\},w)$,
	\item $\mathcal{I}_v = (G-v, k-w(v), F \cup \left\{ s, t, u \right\},w)$.
	\end{itemize}
	If an instance~$\mathcal{I}_x$ was derived by removing a permanent vertex~$x \in \left\{ s, t, u, v \right\} \cap F$, then do not branch on~$\mathcal{I}_x$.
\end{brule}

\begin{lemma}
	\label{lem:branch-correct}
	\cref{br:2cvd_overlap} is correct.
\end{lemma}
{
\begin{proof}
	If~$G$ is already a 2-club cluster graph, then there is nothing to do. 
	Otherwise, assume that~$G$ contains~$stuv$, a restricted~$P_4$.
	We have to show that~$\mathcal{I}$ is a yes-instance if and only if at least one of the instances~$\mathcal{I}_s, \mathcal{I}_t, \mathcal{I}_u, \mathcal{I}_v$ is a yes-instance.

	\enquote{$\Leftarrow$}:
	Let~$S'$ be a 2-club vertex deletion set for one of the four instances where the vertex~$x \in \left\{ s, t, u, v \right\}$ was deleted. 
	Because all four instances have a set of permanent vertices that is a superset of~$F$, the set~$S := S' \cup \left\{ x \right\}$ is a 2-club vertex deletion set for~$\mathcal{I}$ unless~$x$ is permanent, in which case the instance~$\mathcal{I}_x$ would have been skipped. 

	\enquote{$\Rightarrow$}:
	If~$\mathcal{I}$ is a yes-instance, then at least one of the vertices~$s, t, u,$ or~$v$ has to be deleted. 
	Let~$S$ be a 2-club vertex deletion set for~$G$ with~$w(S) \leq k$ and~$S \cap F = \emptyset$. 
	If~$s \in S$, then~$\mathcal{I}_s$ is clearly a yes-instance; otherwise if~$t \in S$, then~$\mathcal{I}_t$ is a yes-instance; otherwise if~$u \in S$, then~$\mathcal{I}_u$ is a yes-instance; otherwise~$v \in S$ and~$\mathcal{I}_v$ is a yes-instance.
 \end{proof}
}

An additional feature of \cref{br:2cvd_overlap} is that we can skip branches where a permanent vertex would have been deleted.
Combining \cref{br:2cvd_overlap} with the fact that one can find a restricted~$P_4$ by running a breadth-first-search from each vertex, we arrive at the following.

\begin{proposition}\label{prop:4^k-searchtree}
	\cstrtwoclubvertexdeletelong{} can be solved in~$O(4^k \cdot nm)$ time.
\end{proposition}

\subsubsection{Data Reduction Rules} \label{sect:2cvd_rrules}
Data reduction rules may be considered the most valuable contribution of parameterized algorithmics to algorithm engineering~\cite{KNN15}.
In this section, we will introduce polynomial-time data reduction rules that can be applied in each step of our search tree algorithm.
We categorize these rules into two types: 
The first type removes vertices and shrinks the graph. 
The second type increases the set of permanent vertices, which in turn can trigger data reduction rules of the first type, and decreases the number of branches

\paragraph*{Shrinking the graph.}
We start with describing rules of the first type.
First note that we can always safely remove a connected component of the graph that is already a 2-club. 
\begin{rrule}\label{rr:2cvd_remove_2clubs}
If~$G$ contains aconnected  component~$C$ that is a 2-club, then delete all vertices in~$C$.
\end{rrule}

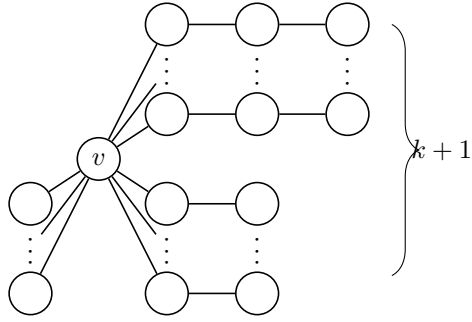
\begin{figure}[t]
	\centering
	\begin{subfigure}[t]{.45\textwidth}
		\centering
		\begin{tikzpicture}
			\node[alter] (p-11) at (0,0) {};
			\node[alter, right=4ex of p-11] (p-12) {};
			\node[alter, right=4ex of p-12] (p-13) {};

			\draw[majarr] (p-11) edge (p-12);
			\draw[majarr] (p-12) edge (p-13);

			\node[alter, below=4ex of p-11] (p-21) {};
			\node[alter, right=4ex of p-21] (p-22) {};
			\node[alter, right=4ex of p-22] (p-23) {};
			
			\draw[majarr] (p-21) edge (p-22);
			\draw[majarr] (p-22) edge (p-23);

			\node[alter, below=4ex of p-21] (p-32) {};
			\node[alter, left =8ex of p-32] (p-31) {};
			\node[alter, right=4ex of p-32] (p-33) {};
			
			\draw[majarr] (p-32) edge (p-33);

			\node[alter, below=4ex of p-32] (p-42) {};
			\node[alter, left =8ex of p-42] (p-41) {};
			\node[alter, right=4ex of p-42] (p-43) {};
			
			\draw[majarr] (p-42) edge (p-43);

			\coordinate (m1) at ($(p-11)!0.5!(p-42)$);
			\coordinate (m2) at ($(p-31)!0.5!(p-32)$);
			\coordinate (m3) at ($(m1)!0.5!(m2)$);

			\node[alter, left=4ex of m1] (v) {$v$};
			\draw[majarr] (v) edge (p-11);
			\draw[majarr] (v) edge (p-21);
			\draw[majarr] (v) edge (p-31);
			\draw[majarr] (v) edge (p-32);
			\draw[majarr] (v) edge (p-41);
			\draw[majarr] (v) edge (p-42);

			\node[rotate= 90] (p-1d) at ($(p-11)!0.5!(p-21)$) {\dots};
			\node[rotate= 90] (p-2d) at ($(p-12)!0.5!(p-22)$) {\dots};
			\node[rotate= 90] (p-3d) at ($(p-13)!0.5!(p-23)$) {\dots};
			\node[rotate= 90] (p-4d) at ($(p-31)!0.5!(p-41)$) {\dots};
			\node[rotate= 90] (p-5d) at ($(p-32)!0.5!(p-42)$) {\dots};
			\node[rotate= 90] (p-6d) at ($(p-33)!0.5!(p-43)$) {\dots};

			\draw[majarr] (v) edge (p-1d);
			\draw[majarr] (v) edge (p-4d);
			\draw[majarr] (v) edge (p-5d);

			\coordinate[right=1ex of p-13] (b1);
			\coordinate[below=22ex of b1] (b2);
			\draw [decorate,decoration={brace,amplitude=10pt,raise=4pt},yshift=0pt]
				(b1) -- (b2) node [black,midway,xshift=0.8cm] {$k+1$};
		\end{tikzpicture}
	\end{subfigure}
\caption{A graph with~$k+1$~restricted~$P_4$'s that overlap only in vertex~$v$. 
}
\label{fig:k_overlapping_p4}
\end{figure}

Now let us consider a graph with~$k+1$ restricted~$P_4$'s that only overlap in one vertex~$v$ as shown in \cref{fig:k_overlapping_p4}. 
We clearly have to delete~$v$.
This basic observation can be generalized to our weighted case. 
\begin{rrule}\label{rr:w2cvd_k_disjoint}
	Let~$\mathcal{P}$ be a set of restricted~$P_4$'s that each contain the vertex~$v$ and such that each vertex~$u \in V$ other than~$v$ is contained in at most~$w(u)$ many restricted~$P_4$'s in~$\mathcal{P}$. 
	If~$|\mathcal{P}| \geq k+1$, then delete~$v$ and decrease~$k$ by~$w(v)$.
\end{rrule}
\begin{lemma}
\cref{rr:w2cvd_k_disjoint} is safe.
\end{lemma}
\begin{proof}
	Let~$v$ and~$\mathcal{P}$ be as above. 
	Clearly, deleting~$v$ would eliminate all restricted~$P_4$'s in~$\mathcal{P}$. 
	Let~$u \in V$ be some other vertex. 
	Denote by~$\ell$ the number of~$P_4$'s in~$\mathcal{P}$ that contain~$u$, which means that deleting~$u$ eliminates~$\ell$~$P_4$'s in~$\mathcal{P}$. The cost of deleting~$u$ is~$w(u) \geq \ell$. This means that eliminating all~$P_4$'s in~$\mathcal{P}$ without deleting~$v$ has a cost of at least~$|\mathcal{P}| \geq k+1$, which is not possible with a budget of~$k$.
 \end{proof}

For the next data reduction rule, we consider \emph{twins}, that is, two vertices with either the same closed neighborhood or the same open neighborhood.    
\begin{observation} \label{thm:2club_with_twins}
If a graph is a 2-club and contains two twins~$u, v$, then after deleting one of them, the graph remains a 2-club. Likewise a 2-club to which a twin is added remains a 2-club. 
\end{observation}
Due to \cref{thm:2club_with_twins} we can delete one of the twins, because they are always in the same 2-club in an optimal solution. 
\begin{rrule}\label{rr:2cvd_twin}
Given two vertices~$u, v \in V$ such that either~$N[u] = N[v]$ or~$N(u) = N(v)$, %
delete~$v$ and set~$w(u)$ to~$w(u) + w(v)$.
\end{rrule}

\begin{lemma}
\label{lem:twins}
\cref{rr:2cvd_twin} is safe.
\end{lemma}

\begin{proof}
We have to show that~$(G, w, k)$ is a yes-instance if and only if~$(G', w', k)$ is a yes-instance.

\enquote{$\Leftarrow$}: 
Let~$S'$ be an optimal 2-club vertex deletion set with $ w'(S') \leq k$.
If~$S'$ removes~$u$, then~$S = S' \cup \left\{ v \right\}$ is a 2-club vertex deletion set for~$G$ with~$ w'(S') = w(S)$. Otherwise, by \cref{thm:2club_with_twins}, $S'$ is also a 2-club vertex deletion set for~$G$.

\enquote{$\Rightarrow$}: 
Let~$S$ be an optimal 2-club vertex deletion set with~$w(S)\leq k$ and~$H = G[V \setminus S]$. We claim that~$S$ either removes both~$u$ and~$v$ or neither of them. Assume without loss of generality that~$S$ only removes~$u$. Then by \cref{thm:2club_with_twins} the set~$S' = S \setminus \{u\}$ would be a 2-club vertex deletion set as~$u$ and~$v$ would be twins in~$H' = G[V \setminus S']$. Since~$S'$ is smaller than~$S$, this is a contradiction to the optimality of~$S$.
 \end{proof}

If a restricted~$P_4$~$stuv$ has to be eliminated, then one has to potentially branch into four cases. 
If some of the vertices of~$stuv$ are permanent, then this decreases this number branches. 
In the cases that three vertices of~$stuv$ are permanent, then there is only one possible branch.
\begin{rrule}\label{rr:2cvd_simple_choice}
If the graph~$G$ contains a restricted~$P_4$~$stuv$ such that only one vertex~$x \in \left\{ s,t,u,v \right\}$ is not permanent, then delete~$x$ and decrease~$k$ by~$w(x)$.
\end{rrule}

The next rule can be used to shrink some subgraphs in which all vertices are permanent.

\begin{rrule} \label{rr:2cvd_shrink_permanent}
	Let~$v$ be a vertex in~$G$ that is not permanent. 
	Let~$C$ be a component of~$G-v$ that is a 2-club and all its vertices are marked permanent.
	Let~$d$ be the maximum distance of a vertex in~$C$ to~$v$. 
	Replace~$C$ with~$d$ new permanent vertices that together with~$v$ induce a path of length~$d$.
\end{rrule}
\begin{lemma}
\label{lem:shrink-permanent}
\cref{rr:2cvd_shrink_permanent} is safe.
\end{lemma}
\begin{proof}
Let~$C$ and~$v$ be as in \cref{rr:2cvd_shrink_permanent}.
Note that the vertices in~$C$ are all marked permanent. If any vertex $u \in C$ is part of a restricted~$P_4$, then~$v$ must also be part of the restricted~$P_4$. This means that only the distance between~$u$ and~$v$ is important and the path with at most~$d+1$ vertices that starts in~$v$ is sufficient to represent all vertices that have some distance (which is at most three) to~$v$.
 \end{proof}

\paragraph*{Increasing the set of permanent vertices.}
Note that \cref{rr:2cvd_simple_choice,rr:2cvd_shrink_permanent} both require permanent vertices to trigger. 
So far, the only way to add new permanent vertices is with \cref{br:2cvd_overlap}.
Subsequently, we discuss a faster way of producing permanent vertices via data reduction rules.

For the following data reduction rule, we need the concept of a \emph{bridge vertex}.
We call a vertex~$b$ a bridge vertex if for some~$s, v \in N(b)$ there exists an induced~$P_4$~$stuv$ for some~$t,u \in V$. 
We say that~$b$ bridges~$stuv$. 

A vertex~$v \in V$ is a \emph{cut vertex} if the deletion of~$V$ increases the number of connected components.
If a 2-club $C$ becomes isolated by removing a cut vertex~$v$, then the vertices of~$C$ can be safely excluded from the solution (by marking them as permanent) under the following premise: 
If any vertex from that 2-club~$C$ was included in the solution, then it can be replaced by~$v$ to yield another solution of the same size.
This premise holds if $v$ is not a bridge vertex and if~$w(v) \leq \min_{u \in C} w(u)$. 
For an example of the application of the corresponding~\cref{rr:permanent2club} see \cref{subfig:permanent2club}.

\begin{rrule} \label{rr:permanent2club}
Given a cut vertex~$v$ in~$G$ that is not a bridge vertex and not permanent. For any component~$C$ in~$G - v$, if~$C$ is a 2-club and $w(v) \leq \min_{u \in C} w(u)$, then mark the vertices in~$C$ as permanent.

\begin{figure}[t]
	\centering
	\begin{subfigure}[c]{.45\textwidth}
		\centering
		\begin{tikzpicture}

			\def\n{6}
			\foreach \i/\t in {1/,2/,3/$a$,4/$v$,5/$b$,6/} {
				\node[alter] (A-\i) at ({360 * (\i-1) / \n}:8ex) {\t};
			}
			\foreach \i [evaluate={\j=int(mod(\i,\n)+1}] in {1,...,\n} {
				\draw[majarr] (A-\i) edge (A-\j);
			}

			\begin{scope}[shift={ (-18ex, 0) }]
				\coordinate (C-center) at (0,0);

				\foreach \i in {1,...,3} {
					\node[alter] (C-\i) at ({360 * (\i-1) / 3}:3.5ex) {};
				}
			\end{scope}

			\draw[majarr] (A-4) edge (C-1);
			\draw[majarr] (C-2) edge (C-1);
			\draw[majarr] (C-3) edge (C-1);

			\begin{pgfonlayer}{background}
				\filldraw[fill=gray!10,rounded corners=15mm, draw=black!50, dashed,semithick]
					(C-center) circle (6.5ex);
			\end{pgfonlayer}

		\end{tikzpicture}
		\caption{\cref{rr:permanent2club} can be applied on vertex~$v$. The gray area contains vertices that can be marked as permanent.}
		\label{subfig:permanent2club}
	\end{subfigure}~~~~
	\begin{subfigure}[c]{.45\textwidth}
		\centering
		\begin{tikzpicture}

			\def\n{5}
			\foreach \i/\t in {1/$v$,2/$b$,3/,4/,5/$a$} {
				\node[alter] (A-\i) at ({360 * (\i-1) / \n + 180}:6ex) {\t};
			}
			\foreach \i [evaluate={\j=int(mod(\i,\n)+1}] in {1,...,\n} {
				\draw[majarr] (A-\i) edge (A-\j);
			}

			\begin{scope}[shift={ (-16ex, 0) }]
				\coordinate (C-center) at (0,0);

				\foreach \i/\t in {1/$c$,2/,3/} {
					\node[alter] (C-\i) at ({360 * (\i-1) / 3}:3.5ex) {\t};
				}
			\end{scope}
			\begin{pgfonlayer}{background}
				\filldraw[fill=gray!10,rounded corners=15mm, draw=black!50, dashed,semithick]
					(C-center) circle (6.5ex);
			\end{pgfonlayer}

			\draw[majarr] (A-1) edge (C-1);
			\draw[majarr] (C-2) edge (C-1);
			\draw[majarr] (C-3) edge (C-1);
		\end{tikzpicture}
		\caption{\cref{rr:permanent2club} cannot be applied on vertex~$v$ because $v$ is a bridge for~$a$ and~$b$.}
		\label{subfig:nopermanent2club}
	\end{subfigure}
\caption{Examples of graphs with a cut vertex~$v$ (all vertex weights are one). 
		The gray area is a 2-club that is isolated from the rest of the graph after deleting~$v$. 
		Note that in the left graph the removal of~$v$ increases the distance of~$a$ and~$b$ to five.
		Thus no induced~$P_4$ exists between~$a$ and~$b$. 
}
\label{fig:2cvd_permanent2club_examples}
\end{figure}
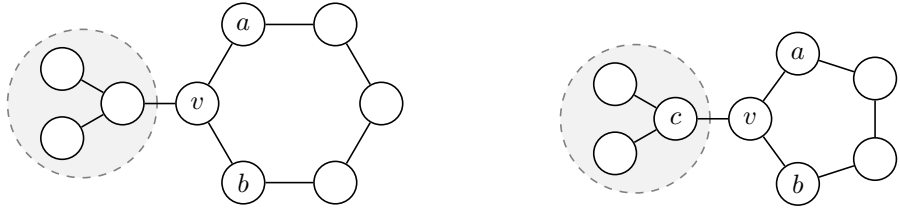
\end{rrule}
\begin{lemma}
\cref{rr:permanent2club} is safe.
\end{lemma}
\begin{proof}
	If~$C$ is a 2-club, then it contains no restricted~$P_4$. 
	If any vertex in~$C$ is part of a restricted~$P_4$, then~$v$ must be on this path. 
	Suppose that an optimal solution~$S$ removes some vertices~$F \subseteq C$, $F \neq \emptyset$. 
	Removing~$F$ is only needed to eliminate~$P_4$'s that start in~$C$, which also necessarily contain~$v$. 
	If~$v \in S$, then~$S$ is not optimal. 
	Otherwise, we claim that~$S' = (S \setminus F) \cup \{v\}$ is another optimal solution. 
	This holds true because~$v$ eliminates the same~$P_4$'s as~$F$, has weight~$w(v) \leq  w(F)$, and because~$v$ is not a bridge vertex. 
	Hence it follows that removing~$v$ cannot contribute to the creation of a restricted~$P_4$.
 \end{proof}
For an example why we require~$v$ to be a non-bridge vertex in \cref{rr:permanent2club}, see \cref{subfig:nopermanent2club}.

\cref{rr:permanent2club} can be slightly generalized. 
To this end, we define the \emph{robustness} of two vertices $s$ and $v$ as the number of vertices that need to be deleted before a restricted~$P_4$ can be created between~$s$ and~$v$:
$$ \robust(s, v) := \begin{cases} \infty & \text{if there is no induced~$P_4$~$stuv$ for any~$t,u \in V$,} \\ w(N(s) \cap N(v))   & \text{otherwise.} \end{cases}$$
For an example see \cref{fig:2cvd_robustness}.
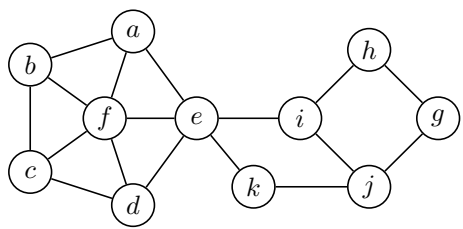
\begin{figure}[t]
	\centering
	\begin{subfigure}[c]{.45\textwidth}
		\centering
		\begin{tikzpicture}

			\def\n{5}
			\foreach \i/\t in {1/$e$,2/$a$,3/$b$,4/$c$,5/$d$} {
				\node[alter] (A-\i) at ({360 * (\i-1) / \n}:8ex) {\t};
			}
			\node[alter] (v) at (0,0) {$f$};
			\foreach \i [evaluate={\j=int(mod(\i,\n)+1}] in {1,...,\n} {
				\draw[majarr] (A-\i) edge (A-\j);
				\draw[majarr] (A-\i) edge (v);
			}

			\begin{scope}[shift={ (23ex, 0) }]
				\coordinate (C-center) at (0,0);

				\foreach \i/\t in {1/$g$,2/$h$,3/$i$,4/$j$} {
					\node[alter] (C-\i) at ({360 * (\i-1) / 4}:6ex) {\t};
				}
				\foreach \i [evaluate={\j=int(mod(\i,4)+1}] in {1,...,4} {
					\draw[majarr] (C-\i) edge (C-\j);
				}
			\end{scope}

			\node[alter, left = 6.2ex of C-4] (u) {$k$};

			\draw[majarr] (A-1) edge (u);
			\draw[majarr] (A-1) edge (C-3);
			\draw[majarr] (C-4) edge (u);
		\end{tikzpicture}
	\end{subfigure}
\caption{An example to demonstrate robustness (all vertex weights are one). For instance, we have $\robust(a, d) = 2$ because $f$ and $e$ have to be deleted to ``promote'' the induced $P_4$ $abcd$ to a restricted one. Other robustness values are: $\robust(f, j) = \robust(h, k)= 0$ and~$\robust(h, j) = \infty$. Two optimal 2-club cluster vertex deletion sets are~$\left\{  e,j \right\}$ and~$\left\{ i,k \right\}$.}
\label{fig:2cvd_robustness}
\end{figure}
An induced~$P_4$~$stuv$ is a restricted~$P_4$ if and only if~$\robust(s, v) = 0$. For the induced~$P_4$~$stuv$ the set~$U := N(s) \cap N(v)$ of vertices needs to be deleted, before~$stuv$ is \enquote{promoted} to a restricted~$P_4$. We will say that the deletion of the vertices in~$U$ \emph{contributes to the creation} of the restricted~$P_4$~$stuv$.

In \cref{rr:permanent2club} the deletion of~$v$ cannot decrease the robustness of any two vertices because~$v$ is not a bridge vertex. However, we do not want to restrict ourselves to just vertices that cannot decrease robustness. We adapt this rule to also allow~$v$ to be a bridge vertex, but we still need to guarantee that the deletion of~$v$ cannot contribute to the creation of a restricted~$P_4$. For this we consider our remaining budget~$k$ and conclude that if the robustness in the neighborhood is sufficiently high, then we can still mark the 2-clubs as permanent under the same premise that~$v$ can be deleted instead of any vertex in those 2-clubs. Additionally, we do not need to consider how removing~$v$ affects the robustness between vertices in 2-clubs that would be isolated, because we already know that they are 2-clubs and do not have restricted~$P_4$'s.

\begin{rrule} \label{rr:permanent2club_v2}
	Given a vertex~$v$ in~$G$ that is not permanent. 
	Let~$C_1, \dots, C_{\ell}$ be the components of~$G-v$ that are 2-clubs, and~$H = G - v -C_1 - \dots - C_{\ell}$. If for all pairs of vertices~$a, b \in N_G(v) \cap V(H)$~$\robust_G(a, b) > k$ and for all 2-club components $C_i$, $i \in \{1,\ldots,\ell\}$, we have $w(v) \leq \min_{u \in C_i} w(u)$, then mark all vertices in~$C_1, \dots, C_{\ell}$ as permanent.
\end{rrule}
\begin{lemma}
\label{lem:permanent2club_multiple}
\cref{rr:permanent2club_v2} is safe.
\end{lemma}

\begin{proof}
	If~$C_i$ is a 2-club, then it contains no restricted~$P_4$. 
	If any vertex in~$C_i$ is part of a restricted~$P_4$, then~$v$ must be on this path. 
	Suppose that an optimal solution~$S$ removes some vertices~$F \subseteq C_i$ with~$F \neq \emptyset$. 
	Removing~$F$ is only needed to eliminate~$P_4$'s that start in~$C_i$, which also necessarily contain~$v$. 
	If~$v \in S$, then~$S$ is not optimal. 
	Otherwise we claim that~$S' = (S \setminus F) \cup \{v\}$ is another optimal solution.
	Deleting~$v$ reduces the robustness between its neighbors. 
	However, deleting~$v$ cannot create a restricted~$P_4$. 
	If deleting~$v$ created a restricted~$P_4$, then this~$P_4$ would need to start and end in two neighbors of~$v$. 
	Deleting~$v$ cuts off the 2-clubs~$C_1, \dots, C_{\ell}$ which means no restricted~$P_4$'s were created in them, which means the neighbors of~$v$ in these 2-clubs need not be considered further. 
	The only other vertices that could be affected are those in~$U = N_G(v) \cap V(H)$. 
	Because the pairwise robustness of vertices in~$U$ is at least~$k+1$, this means that vertices with a total budget greater than $k$ need to be removed before there can be a restricted~$P_4$ that starts and ends in~$U$. 
	Because the budget~$k$ does not allow that to happen, replacing~$F$ by the single vertex~$v$ with $w(v) \leq \sum_{u\in F}w(u)$ to obtain~$S'$ results in another optimal solution.
 \end{proof}

A 2-club vertex deletion set~$S$ is clearly not optimal if a vertex~$v$ can be removed from it and~$S' = S \setminus \left\{ v \right\}$ remains a 2-club vertex deletion set.
\begin{observation} \label{obs:2cvd_optimality}
Let~$S$ be a 2-club vertex deletion set of~$G$. If~$N[v] \subseteq S$ for some~$v \in V$, then~$S \setminus \{v\}$ is also a 2-club vertex deletion set.
\end{observation}
One could use \cref{obs:2cvd_optimality} as a simple test whether a 2-club vertex deletion set~$S$ is minimal. Clearly we do not have to wait until we have found a 2-club vertex deletion set~$S$ to apply this test. A \emph{partial 2-club vertex deletion set}~$S'$ is a set of vertices which were removed from~$G$ along the way from the root to a branching node of the search tree. The test can be applied to~$S'$ in the same way as if it would be applied to~$S$. Additionally, if the removal of any vertex in~$G$ would cause this test to fail, then this vertex must not be removed.
\begin{rrule} \label{rr:2cvd_prevent_inoptimality}
Let~$S'$ be a partial 2-club vertex deletion set of~$G$ constructed at some stage of the branching. If for any~$v \in S'$ we have $|N(v) \setminus S'| = 1$, then mark the unique vertex~$x \in N(v) \setminus S'$ as permanent.
\end{rrule}

\begin{lemma}
	\label{lem:permanent_vertex}
	\cref{rr:2cvd_prevent_inoptimality} is safe.
\end{lemma}

\begin{proof}
	Let~$S'$ and~$x$ be as above. 
	Any 2-club vertex deletion set~$S$ with~$(S' \cup \{x\}) \subseteq S$ is not a minimal deletion set by \cref{obs:2cvd_optimality}. 
	This implies that~$x$ cannot be in any minimal solution containing $S'$ and hence we can mark $x$ as permanent.
 \end{proof}

\subsubsection{Lower Bounds.}\label{ssec:lowerbounds}
Another way to shrink the size of a search tree are \emph{lower bounds}.
A lower bound can be thought of as a function~$\ell(G)$ of the graph~$G$ such that~$\ell(G) \leq |S|$, where~$S$ is an optimal 2-club vertex deletion set for~$G$. Lower bounds are a very practical way to shrink the size of the search tree, because if once for the solution size parameter~$k$ it holds that~$k < \ell(G)$, then we know that there is no solution.

\begin{lbound}\label{lb:w_disjoint_bound}
Let~$\mathcal{P} = \left\{ p_1, \dots, p_{\ell} \right\}$ be a set of restricted~$P_4$'s in~$G$ such that each vertex~$v \in V$ is contained in at most~$w(v)$ many restricted~$P_4$'s in~$\mathcal{P}$. Then a minimum vertex deletion set for~$G$ has size at least~$\ell$.
\end{lbound}
{
\begin{proof}
Let~$\mathcal{P}$ be as above and~$v \in V$ be any vertex. Denote by~$r$ the number of~$P_4$'s in~$\mathcal{P}$ that contain~$v$, which means deleting~$v$ eliminates~$r$~$P_4$'s in~$\mathcal{P}$. The cost of deleting~$v$ is~$w(v) \geq r$. This means eliminating all~$P_4$'s in~$\mathcal{P}$ has a cost of at least~$|\mathcal{P}|$. Clearly, if a restricted~$P_4$ in~$\mathcal{P}$ is not eliminated, then the graph is not a 2-club cluster graph. 
 \end{proof}
}

Next, we exploit the size of a minimum vertex cut set. A vertex cut set~$D$ of a graph~$G$ is a set of vertices such that $G-D$ is disconnected.  
We know that an optimal 2-club vertex deletion set~$S$ for~$G$ splits it into multiple 2-clubs, which means that~$S$ must be a vertex cut set. 
However, an optimal 2-club vertex deletion set is not always a vertex cut set. 
The following lower bound overcomes this problem:

\begin{lbound}\label{lb:cut_bound}
Let~$G$ be a connected graph, let~$C$ be the maximum-weight 2-club in~$G$, and let~$D$ be the minimum-weight vertex cut set of~$G$. Then a minimum-weight 2-club vertex deletion set for~$G$ has size at least~$\min(w(V\setminus C), w(D))$.
\end{lbound}
{
\begin{proof}
Let~$S$ be an optimal 2-club vertex deletion set and let~$G' := G[V \setminus S]$ be the resulting 2-club cluster graph. If~$G'$ contains at least two components, then~$S$ is a vertex cut set for~$G$ and~$w(S) \geq w(D)$. If~$G'$ has only one component, then~$V \setminus S = V(G')$ is the maximum-weight 2-club in~$G$ and~$w(S) = w(V \setminus  C)$.
 \end{proof}
}

\section{Implementation of the Branch\&Bound Algorithm} \label{sec:implementation}

In this section, we discuss some implementation details of our algorithm for \cstrtwoclubvertexdelete{} that have been left open in \cref{app:sec:BB}. 
For example we did not say how to compute a set of restricted~$P_4$'s that can overlap in complex ways. 
It is clear that ideally we would like to find such a set whose size is maximum. 
This is likely an NP-hard problem (refer to \citet{IPS82}), which is why in a practical implementation we would rather have a fast heuristic that offers good results in most cases. 
We will now discuss heuristics used in our solver (see \cref{sec:impl-exp} for details) and other implementation challenges of interest.

\paragraph*{Determining the Minimum 2-Club Vertex Deletion Set Size.}
In order to find the minimum 2-club vertex deletion set size we simply try increasing values for the budget~$k$, as can be seen in \cref{alg:2cvd_increment_k}. 
Naturally, our solver also outputs the 2-club vertex deletion set that was found.

\begin{algorithm}[t]
	\begin{algorithmic}[1]
	\caption{Finding a minimum 2-club vertex deletion set} \label{alg:2cvd_increment_k}
	\Input{An undirected graph~$G = (V, E)$}
	\Output{The minimum 2-club vertex deletion set of~$G$}
	\\
		initialize the weight function~$w$ with~$w(v) = 1$ for each~$v \in V$ \\
		$(G', \infty, F', w'$) $\leftarrow$ apply data reduction rules to~$(G, k = \infty, F = \emptyset, w)$ \\
		$k \leftarrow \textsc{lower bound}(G', w', F')$ \\
		\textbf{while} $(G', k, F' , w')$ \text{ is a no-instance of \weightedtwoclubvertexdeleteshort{}} \\
			\hspace*{10mm} $k \leftarrow k + 1$ 
		\State \textbf{return} found solution of size~$k$
	\end{algorithmic}
\end{algorithm}

\paragraph*{\cref{br:2cvd_overlap}.} 
This is the branching rule that allows us to mark some vertices as permanent and to skip branches in which a permanent vertex would have been deleted (see \cref{ssec:search-tree}). 
It also allows us to freely choose any restricted~$P_4$ to branch on. 
It is highly advantageous to choose a~$P_4$ that contains permanent vertices because for each permanent vertex we are allowed to skip one out of a total of four branches. 
For this reason we select a restricted~$P_4$ that contains the most permanent vertices, and if there is more than one, then we select the one in which the average weight of non-permanent vertices is the highest because then on average~$k$ is decreased by a larger value in the branches and thus also making the search tree smaller.

\paragraph*{Handling multiple connected components.}
Each connected component can be solved separately. 
However, we do not know how to distribute the budget~$k$ among these components. 
As in \cref{alg:2cvd_increment_k} we try to solve each component with as little budget as possible, first trying small values for~$k$ and then increasing it by one each time. 
An improvement is to sort all components by size and when solving the last component to give it all remaining budget, which prevents us from trying many different~$k$-values for the last component. 
While this only gives an improvement by a constant factor of at most four, the effect is much more noticeable when the graph repeatedly decomposes into one large component and a few smaller ones. 
If from the root of a search tree to some leaf this happens~$i$ many times, then we have a speedup of up to~$4^i$ along those branches of the search tree.

\paragraph*{Disjoint restricted $P_4$'s.} 
\cref{rr:w2cvd_k_disjoint} allows us to delete a vertex~$v$ if there is a set~$\mathcal{P}$ of~$k+1$ restricted~$P_4$'s that each contain~$v$, but otherwise each vertex~$u$ can only be present in at most~$w(u)$ many restricted~$P_4$'s. 
We would like to find a maximum set~$\mathcal{P}$ and then test if its size is at least~$k+1$. 
However, this proved to be quite challenging. 
For our implementation we use a heuristic that does not guarantee finding a maximum set.

We focus on finding a maximum set that only contains those restricted~$P_4$'s that start in vertex~$v$. 
This means we do not try to find~$P_4$'s where~$v$ might be in the \enquote{middle} (see \cref{fig:k_overlapping_p4}). 
This can then be modeled as a simple maximum flow problem. 
The algorithm is described in \cref{alg:kplus_flow}.
\\
\begin{algorithm}[t]
	\begin{algorithmic}[1]
		\caption{Heuristic for \cref{rr:w2cvd_k_disjoint}} \label{alg:kplus_flow}
		\Input{A graph~$G = (V, E)$, with a weight function~$w$, a vertex~$v$ and~$k \in \NN$}
		\Output{\texttt{true} if \cref{rr:w2cvd_k_disjoint} can be applied to remove~$v$ from~$G$}
		\\
			Create a directed graph~$G_{\text{flow}}$ containing only the vertex~$s$ and~$t$ \\
			\textbf{for each} $i \in \left\{ 1, 2, 3 \right\}$ \Comment{Split vertices into layers based on distance}\\
				\hspace*{10mm} $D_i \leftarrow $ all vertices with distance $i$ to $v$  in $G$ \\
			\textbf{for each} $u \in D_1 \cup D_2 \cup D_3$ \Comment{Limit flow through a vertex}\\
				\hspace*{10mm} add the vertex $u_{\text{in}}$ and $u_{\text{out}}$ to $G_{\text{flow}}$ \\
				\hspace*{10mm} add the edge $(u_{\text{in}}, u_{\text{out}})$ with a capacity of $w(u)$ to $G_{\text{flow}}$ \\
			\textbf{for each} $u \in D_1$ \Comment{Connect layer 1 to source}\\
				\hspace*{10mm} add the edge $(s, u_{\text{in}})$ with infinite capacity to $G_{\text{flow}}$ \\
			\textbf{for each} $i \in \left\{ 1, 2 \right\}$ \Comment{Connect the layers}\\
				\hspace*{10mm} \textbf{for each} $u \in D_i$ \\
				\hspace*{10mm} \hspace*{10mm} \textbf{for each} $x \in D_{i+1}$ with $\left\{ u, x \right\} \in E$ \\
				\hspace*{10mm} \hspace*{10mm} \hspace*{10mm} add the edge $(u_{\text{in}}, x_{\text{out}})$ with infinite capacity to $G_{\text{flow}}$ \\
			\textbf{for each} $u \in D_3$ \Comment{Connect layer 3 to sink}\\
				\hspace*{10mm} add the edge $(u_{\text{out}}, t)$ with infinite capacity to $G_{\text{flow}}$ \\
			\smallskip
			$f \leftarrow $ maximum $s$-$t$ flow in $G_{\text{flow}}$
			\If{$f > k$}
				\textbf{return} \texttt{true}
			\EndIf
			\State \textbf{return} \texttt{false}
	\end{algorithmic}
\end{algorithm}
Because for a restricted~$P_4$~$stuv$ the distance from~$v$ to~$u$,~$t$, and~$s$ is one, two and three, respectively, we partition the vertices in the graph into three sets~$D_1, D_2, D_3$. 
No restricted~$P_4$ can contain two vertices from the same~$D_i$, which is why our flow graph only contains edges from~$D_i$ to~$D_{i+1}$. 
We make sure that a vertex~$u$ is part of at most~$w(u)$ many~$P_4$'s by splitting it into two vertices connected by an edge with~$w(u)$ capacity. 
As a result there can be only flow along paths of type~$su_1^{\text{in}}u_1^{\text{out}}u_2^{\text{in}}u_2^{\text{out}}u_3^{\text{in}}u_3^{\text{out}}t$, and sending a flow of value 1 along that path means adding the restricted~$P_4$~$vu_1u_2u_3$ to~$\mathcal{P}$. 
The final maximum flow in~$G_{\text{flow}}$ does not uniquely identify a set~$\mathcal{P}$; however, the value of the maximum flow tells us the size of all maximum size~$\mathcal{P}$'s. 
Because we are only interested in the size of the set~$\mathcal{P}$ this is all we need.

\paragraph*{Disjoint restricted $P_4$'s lower bound.} 
Here we use a set of restricted~$P_4$'s~$\mathcal{P}$ such that each vertex~$v$ is present in at most~$w(v)$ many restricted~$P_4$'s. 
The size of this set is then the lower bound. 
We compute the set~$\mathcal{P}$ using a greedy heuristic. 
Each vertex has a counter initialized with the value of its weight. 
This counter keeps track of how many times this vertex can be used in a restricted~$P_4$. 
We iterate over all vertices in~$V$ by increasing degree and for each vertex~$s$ we look for a restricted~$P_4$~$stuv$ such that for all four vertices of this~$P_4$ their counter is positive. 
The restricted~$P_4$ is not chosen randomly, but rather we select a~$P_4$ that minimizes the sum of degrees of its vertices. 
Finding such a~$P_4$ takes~$\bigO(n+m)$ time. 
The~$P_4$~$stuv$ is then implicitly added to~$\mathcal{P}$ by decrementing the counter for~$s,t,u$ and~$v$ by one. 
If the counter for~$s$ remains positive, then we repeat this step and search for another~$P_4$.

The reason for minimizing the sum of the degrees is that selecting a~$P_4$ which contains many high-degree vertices might overlap with and therefore likely exclude many other restricted~$P_4$'s (see \cref{fig:lb_good_heuristic}).

\begin{figure}[t]
	\centering
	\begin{subfigure}[t]{.45\textwidth}
		\centering
		\begin{tikzpicture}
			\node[alter] (p11) at (0,0) {};
			\node[alter, below = 3ex of p11] (p12) {};
			\node[alter, below = 3ex of p12] (p13) {};
			\node[alter, below = 3ex of p13] (p14) {};

			\node[alter, right = 3ex of p11] (p21) {};
			\node[alter, below = 3ex of p21] (p22) {};
			\node[alter, below = 3ex of p22] (p23) {};
			\node[alter, below = 3ex of p23] (p24) {};

			\node[alter, right = 3ex of p21] (p31) {};
			\node[alter, below = 3ex of p31] (p32) {};
			\node[alter, below = 3ex of p32] (p33) {};
			\node[alter, below = 3ex of p33] (p34) {};

			\node[alter, right = 3ex of p31] (p41) {};
			\node[alter, below = 3ex of p41] (p42) {};
			\node[alter, below = 3ex of p42] (p43) {};
			\node[alter, below = 3ex of p43] (p44) {};

			\draw[majarr] (p11) edge (p12) (p12) edge (p13) (p13) edge (p14);
			\draw[majarr] (p21) edge (p22) (p22) edge (p23) (p23) edge (p24);
			\draw[majarr] (p31) edge (p32) (p32) edge (p33) (p33) edge (p34);
			\draw[majarr] (p41) edge (p42) (p42) edge (p43) (p43) edge (p44);

			\draw[majarr] (p12) edge (p22) (p22) edge (p32) (p32) edge (p42);

			\begin{pgfonlayer}{background}
				\draw[rounded corners=1mm, draw=gray, opacity=0.2, line width=22pt, line cap=round]
					(p12.center) -- (p22.center) -- (p32.center) -- (p42.center);
			\end{pgfonlayer}
		\end{tikzpicture}
	\end{subfigure}
	\begin{subfigure}[t]{.45\textwidth}
		\centering
		\begin{tikzpicture}
			\node[alter] (p11) at (0,0) {};
			\node[alter, below = 3ex of p11] (p12) {};
			\node[alter, below = 3ex of p12] (p13) {};
			\node[alter, below = 3ex of p13] (p14) {};

			\node[alter, right = 3ex of p11] (p21) {};
			\node[alter, below = 3ex of p21] (p22) {};
			\node[alter, below = 3ex of p22] (p23) {};
			\node[alter, below = 3ex of p23] (p24) {};

			\node[alter, right = 3ex of p21] (p31) {};
			\node[alter, below = 3ex of p31] (p32) {};
			\node[alter, below = 3ex of p32] (p33) {};
			\node[alter, below = 3ex of p33] (p34) {};

			\node[alter, right = 3ex of p31] (p41) {};
			\node[alter, below = 3ex of p41] (p42) {};
			\node[alter, below = 3ex of p42] (p43) {};
			\node[alter, below = 3ex of p43] (p44) {};

			\draw[majarr] (p11) edge (p12) (p12) edge (p13) (p13) edge (p14);
			\draw[majarr] (p21) edge (p22) (p22) edge (p23) (p23) edge (p24);
			\draw[majarr] (p31) edge (p32) (p32) edge (p33) (p33) edge (p34);
			\draw[majarr] (p41) edge (p42) (p42) edge (p43) (p43) edge (p44);

			\draw[majarr] (p12) edge (p22) (p22) edge (p32) (p32) edge (p42);

			\begin{pgfonlayer}{background}
				\draw[rounded corners=1mm, draw=gray, opacity=0.2, line width=22pt, line cap=round]
					(p11.center) -- (p12.center) -- (p13.center) -- (p14.center);
				\draw[rounded corners=1mm, draw=gray, opacity=0.2, line width=22pt, line cap=round]
					(p21.center) -- (p22.center) -- (p23.center) -- (p24.center);
				\draw[rounded corners=1mm, draw=gray, opacity=0.2, line width=22pt, line cap=round]
					(p31.center) -- (p32.center) -- (p33.center) -- (p34.center);
				\draw[rounded corners=1mm, draw=gray, opacity=0.2, line width=22pt, line cap=round]
					(p41.center) -- (p42.center) -- (p43.center) -- (p44.center);
			\end{pgfonlayer}
		\end{tikzpicture}
	\end{subfigure}
	\caption{
		Two maximal sets of disjoint restricted~$P_4$'s in the same graph. 
		Each~$P_4$ is represented by a gray path. 
		The restricted~$P_4$ on the left contains the highest degree vertices, each of which could have been in its own disjoint restricted~$P_4$ like on the right. 
		Our heuristic for \cref{lb:w_disjoint_bound} (see \cref{ssec:lowerbounds}) prevents such a bad case.
	}
\label{fig:lb_good_heuristic}
\end{figure}
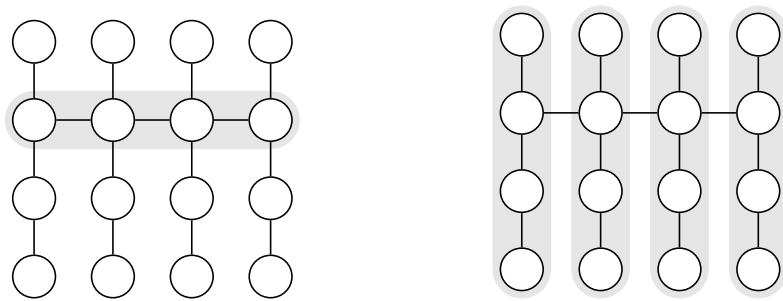

\section{Experimental Evaluation} \label{sec:impl-exp}

In this section, we present experimental results for our \cstrtwoclubvertexdelete{} solver.

\subsection{Setup}

We implemented our branch\&bound algorithm (see \cref{sec:algorithms}) for \weightedtwoclubvertexdelete{} in C++ (we use the algorithm to solve \twoclubvertexdelete{}).\footnote{The source code is available at \url{https://fpt.akt.tu-berlin.de/software/two-club-editing/two-club-vertex-deletion.zip} and includes the source code for the ILP formulation using CPLEX.} 
This solver (called \texttt{solverALL} in the following) computes a 2-club vertex deletion set size of minimum cost and outputs the solution set. 
It uses all data reduction rules described in \cref{sect:2cvd_rrules} and \cref{lb:w_disjoint_bound}.\footnote{Our implementation of \cref{lb:cut_bound} was far too slow to be of use.}
Note that for the implementation of some data reduction rules and lower bounds we use heuristics; see \cref{sec:implementation} for details.
Due to the relatively high running times these data reduction rules are not applied exhaustively, but rather applied once to each vertex.
We will compare the performance of our solver against the ILP formulation from \cref{sec:algorithms} solved using CPLEX (we will refer to this solver as \texttt{CPLEX}). 
All experiments were run on a machine with an Intel Xeon W-2125 8-core, 4.0 GHz CPU and 256GB of RAM running Ubuntu 18.04. 
We used a recent version of CPLEX, 12.8, for our experiments.
We use mostly default parameters and only set \texttt{mip tolerances mipgap} and \texttt{absmipgap} to zero and enabled \texttt{emphasis numerical}. 
CPLEX can use up to 32 threads by default. 
Even though we had 8 cores available, in our experiments CPLEX usually only used four. 
This is an advantage of CPLEX, because our solver was only written to use a single thread.
Our solver only needs up to 20MB of RAM, but we have seen CPLEX to use even 30GB of RAM.
For the running time measurements of our solver and CPLEX we used wallclock time. 
For running time measurements of CPLEX we excluded the time it takes to build the ILP model, which can have~$\bigO(n^4)$ constraints. 
For instances with 250 vertices this process can take 20 seconds, and sometimes even 60. 
However, in the vast majority of cases, the build time was at most 30\% of the total running time.

\subsection{Dataset}\label{dataset}
For our analysis we used a real-world biological dataset\footnote{The dataset is available at \url{https://bio.informatik.uni-jena.de/data/\#cluster\_editing\_data}} that has been used for the evaluation of \textsc{Weighted Cluster Editing} solvers \cite{Rha+07,Boc+09}. 
The vertices in the graphs represent protein sequences and between each vertex there is an edge whose weight represents some sort of similarity of the proteins. 
The edge weights can be positive or negative.

A graph with weighted edges does not match the input of \twoclubvertexdelete{}. 
\citet{HH15} used the following conversion for their (unweighted) \textsc{Cluster Editing} solver: first sort the edges by descending weight, keep the first~$c$\% of edges for some~$c \in [0, 100]$ and discard their weight. 
We additionally delete degree zero vertices from the graphs. 
\citet{HH15} used the values~$c=33$, $c=50$ and~$c=66$, which is also what we did, and obtained three datasets, which we will refer to as Bio33, Bio50 and Bio66, respectively. 
Our experiment results for bio66 are fairly similar to Bio50, which is why we will only discuss results for Bio33 and Bio50.

The Bio33 and Bio50 datasets each contain 3964 instances. See \cref{fig:bio_dataset_size}
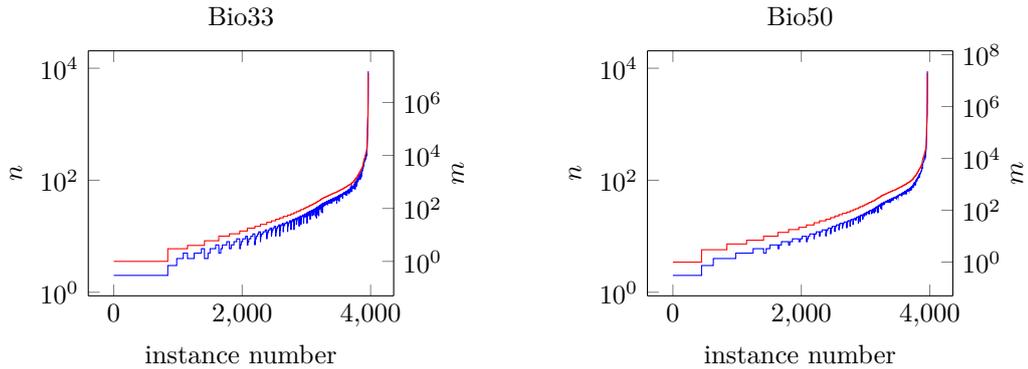
\begin{figure}[t]
	\centering
	\begin{subfigure}[t]{.4\textwidth}
		\centering
		\begin{tikzpicture}[scale=1]
			\begin{axis}[
						width=\textwidth,
						xlabel={instance number},
						ylabel={$n$},
						ylabel near ticks,
						title={Bio33},
						axis y line*=left,
						legend cell align=left,
						legend pos=north west,
						ymode=log
				]
				\addplot[mark=none, blue] table[col sep=comma,y={n}, x={idx}] {instances_bio33.csv};
			\end{axis}
			\begin{axis}[
						width=\textwidth,
						axis y line*=right,
						axis x line=none,
						ylabel={$m$},
						ylabel near ticks,
						ylabel shift = -3pt,
						ytick={1, 100, 10000, 1000000},
						ymode=log,
						ymin=0.05
				]
				\addplot[mark=none,  red] table[col sep=comma,y={m}, x={idx}] {instances_bio33.csv};
			\end{axis}
		\end{tikzpicture}%
	\end{subfigure}
	$\qquad$~$\qquad$
	\begin{subfigure}[t]{.4\textwidth}
		\centering
		\begin{tikzpicture}[scale=1]
			\begin{axis}[
						width=\textwidth,
						xlabel={instance number},
						ylabel={$n$},
						ylabel near ticks,
						title={Bio50},
						axis y line*=left,
						legend cell align=left,
						legend pos=north west,
						ymode=log
				]
				\addplot[mark=none, blue] table[col sep=comma,y={n}, x={idx}] {instances_bio50.csv};
			\end{axis}
			\begin{axis}[
						width=\textwidth,
						axis y line*=right,
						axis x line=none,
						ylabel={$m$},
						ylabel near ticks,
						ylabel shift = -3pt,
						ytick={1, 100, 10000, 1000000, 100000000},
						ymode=log,
						ymin=0.05
				]
				\addplot[mark=none,  red] table[col sep=comma,y={m}, x={idx}] {instances_bio50.csv};
			\end{axis}
		\end{tikzpicture}%
	\end{subfigure}
	\caption{
		Two graphs showing the number of vertices~$n$ in blue (lower line)  and the number of edges~$m$ in red (upper line) in the Bio33 (left) and Bio50 (right) instances (on a log-scale). 
		The instance numbers for the Bio33 and Bio55 instances were selected such that the number of edges increases with the instance number. 
		There are about 15 instances with more than 500 vertices, the largest of which has nearly 9000 vertices.}
	\label{fig:bio_dataset_size}
\end{figure}
for the number of vertices and edges in the instances. 
The \enquote{noise} in the number of vertices is a result of deleting degree-zero vertices from the graphs. 
In \cref{fig:bio_dataset_size} we can see that these datasets contain many instances with less than 50 vertices and a few with around 8000 vertices. 
Our results show that the vast majority of instances with less than 100 vertices can be solved within less than a second. 
For this reason we focused on the harder instances. 
From each dataset we only kept instances with 50--250 vertices. %
After this filtering Bio33 contains 430~instances, whereas Bio50 contains 446 instances.

\subsection{Results}

We next analyze the performance of our solver in detail.
To this end, we start with comparing the theoretical bounds with the results of our experiments.
As can be seen in \cref{fig:impact-n-k}, the number of branches in our search tree is far below the theoretical worst case bound of~$4^k$ (even far below the $3.31^k$ bound of the search tree of Liu 
et al.~\cite{LZZ12}) given in \cref{prop:4^k-searchtree}.
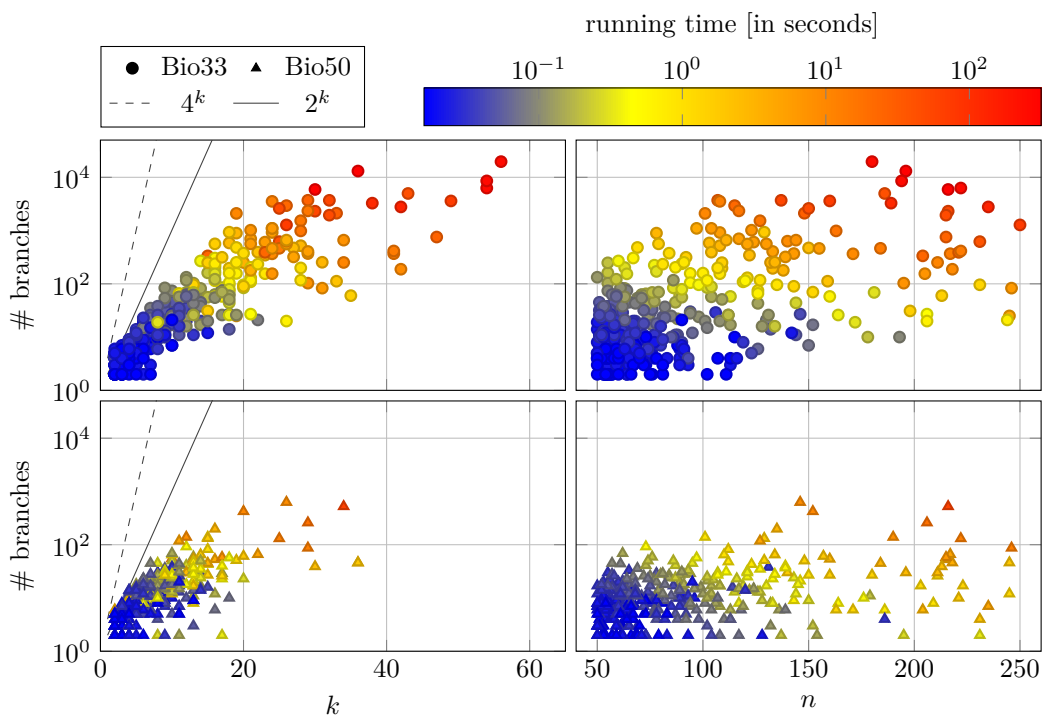
\begin{figure}[!t]
	\centering
	\begin{tikzpicture} 
		\begin{groupplot}[ 
			group style={ 
				group name=my plots, 
				group size=2 by 2, 
				xlabels at=edge bottom, 
				ylabels at=edge left, 
				xticklabels at=edge bottom, 
				yticklabels at=edge left, 
				vertical sep=4pt,
				horizontal sep=4pt
			}, 
			ymode=log,
			grid,
			point meta max=2.5,
			point meta min=-1.8,
			xmin=0,
			xmax=260,
			ymin=1,
			ymax=50000,
			width=0.55\hsize, 
			height=0.35\hsize, 
			ylabel={\# branches},
			cycle multiindex* list = {scatter src=explicit \nextlist scatter \nextlist only marks}, %
		] 
		\nextgroupplot[xmax=65,legend columns=2,
					legend style={
						at={(0,1.05)},
						anchor=south west,
					}]
			\addplot+[mark=*,thick] table[col sep=comma,y={default___steps},x={default___k},meta expr=lg10(\thisrow{default___runtime} + 0.001)] {big_benchmark33.csv};
			\addlegendentry{Bio33}
			\addlegendimage{black,mark=triangle*, only marks}
			\addlegendentry{Bio50}

			\addplot[dashed,color=black!75,domain=1:10,samples=40] {4^x};
			\addlegendentry{$4^k$}
			\addplot[color=black!75,domain=1:20,samples=40] {2^x};
			\addlegendentry{$2^k$}
			\coordinate (top) at (rel axis cs:0,1);%

		\nextgroupplot[xmin=40] 
			\addplot+[mark=*,thick] table[col sep=comma,y={default___steps},x={n},meta expr=lg10(\thisrow{default___runtime} + 0.001)] {big_benchmark33.csv};
		
		\nextgroupplot[xlabel={$k$},xmax=65]
			\addplot+[mark=triangle*,thick] table[col sep=comma,y={default___steps},x={default___k},meta expr=lg10(\thisrow{default___runtime} + 0.001)] {big_benchmark50.csv};
			\addplot[dashed,color=black!75,domain=1:10,samples=40] {4^x};
			\addplot[color=black!75,domain=1:20,samples=40] {2^x};
		
		\nextgroupplot[xlabel={$n$},xmin=40]
			\addplot+[mark=triangle*,thick] table[col sep=comma,y={default___steps},x={n},meta expr=lg10(\thisrow{default___runtime} + 0.001)] {big_benchmark50.csv};
			\coordinate (bot) at (rel axis cs:1,0);%
		\end{groupplot}
		
		\begin{axis}[%
			hide axis,
			scale only axis,
			height=.001\hsize,
			width=0.58\hsize,
			at={(top -| bot)},
			yshift=28pt,
			anchor=south east,
			point meta max=2.5,
			point meta min=-1.8,
			colorbar horizontal,                  %
			colorbar style={
				xtick={-2,-1,0,1,2},
				scaled ticks= true,
				xticklabel pos=upper,
				xlabel={running time [in seconds]},
				xticklabel={ $10^{\pgfmathprintnumber{\tick}}$},
			},
			]
			\addplot [draw=none] coordinates {(0,0)};
		\end{axis}
	\end{tikzpicture} 
	\caption{
		Diagrams illustrating the impact of~$k$ and~$n$ on the running time and number of branches (that is, number of times the branching function is called).
		All four diagrams use the same scale on the $y$-axis.
		The diagrams on top of each other also use the same $x$-axis.
		The values on the $x$-axis are the optimal solution size (left two diagrams) and number of input graph vertices (right two diagrams).
		The top two diagrams show the results for the Bio33-instances; the bottom two for the Bio50 instances.
	}
	\label{fig:impact-n-k}
\end{figure}
This is a clear indication the the data reduction rules and lower bounds have a strong impact in our solver.
Another observation derived from \cref{fig:impact-n-k} is that the impact of the number of input graph vertices on the running time is quite significant. 
The reason for this is the high polynomial running time for computing the data reduction rules and lower bounds:
Our best upper bound on the running time (in terms of~$n$) of one recursive step (including data reduction and lower bounds) is~$O(n^4)$.
One of the bottlenecks in the running time is \cref{rr:w2cvd_k_disjoint}, where we solve up to~$n$ maximum flow instances.
We show subsequently that the high running-time cost for the data reduction rules is justified. 

\cref{fig:impact-n-k} also displays that the Bio33 instances are in general harder for our solver than the Bio50 instances.
The reason for this is that the Bio50 instances are more dense and allow to cluster in less 2-clubs of larger size with fewer vertex removals.

\subparagraph*{Comparisons.}
We next compare our solver \texttt{solverALL} to several variants of it where we deactivate key features and to \texttt{CPLEX}.
The comparisons are illustrated in \cref{fig:comparisons}.
\newcommand{\compareGroupPlotGeneral}[4]{
	\nextgroupplot[#3, ylabel={#1}]
		\addplot+ table[col sep=comma,y={#2___runtime}, x={default___runtime}] {big_benchmark33.csv};
		\addplot[color=black,domain=\minValueRun:\maxValueRun,samples=4] {x};
		\addplot[dashed,color=black!75,domain=\minValueRun:\maxValueRun,samples=4] {5*x};
		\addplot[dashed,color=black!75,domain=\minValueRun:\maxValueRun,samples=4] {0.2*x};
		\addplot[dotted,color=black,domain=\minValueRun:\maxValueRun,samples=4] {25*x};
		\addplot[dotted,color=black,domain=\minValueRun:\maxValueRun,samples=4] {0.04*x};
		\addplot[color=red,mark=none] coordinates {(\minValueRun, 3600) (\maxValueRun, 3600)};
		\addplot[color=red,mark=none] coordinates {(3600, \minValueRun) (3600, \maxValueRun)};

	\nextgroupplot[#4]
		\addplot+ table[col sep=comma,y={#2___runtime}, x={default___runtime}] {big_benchmark33.csv};
		\addplot[color=black,domain=\minValueRun:\maxValueRun,samples=4] {x};
		\addplot[dashed,color=black!75,domain=\minValueRun:\maxValueRun,samples=4] {5*x};
		\addplot[dashed,color=black!75,domain=\minValueRun:\maxValueRun,samples=4] {0.2*x};
		\addplot[dotted,color=black,domain=\minValueRun:\maxValueRun,samples=4] {25*x};
		\addplot[dotted,color=black,domain=\minValueRun:\maxValueRun,samples=4] {0.04*x};
		\addplot[color=red,mark=none] coordinates {(\minValueRun, 3600) (\maxValueRun, 3600)};
		\addplot[color=red,mark=none] coordinates {(3600, \minValueRun) (3600, \maxValueRun)};
}%
\newcommand{\compareGroupPlotTop}[2]{\compareGroupPlotGeneral{#1}{#2}{title={Bio33}}{title={Bio50}}}%
\newcommand{\compareGroupPlot}[2]{\compareGroupPlotGeneral{#1}{#2}{}{}}%
\begin{figure}[t!]
	\centering
	\def\maxValueRun{10000}
	\def\minValueRun{0.0001}

	\begin{tikzpicture}
		\begin{groupplot}[ 
				group style={ 
					group name=my plots, 
					group size=2 by 4, 
					xlabels at=edge bottom, 
					ylabels at=edge left, 
					xticklabels at=edge bottom, 
					yticklabels at=edge left, 
					vertical sep=4pt,
					horizontal sep=4pt
				}, 
				ymode=log,
				grid,
				xmode=log,
				xmin=0.005,
				xmax=6000,
				ymin=0.005,
				ymax=6000,
				ytick={0.01,0.1,1,10,100,1000},
				xtick={0.01,0.1,1,10,100,1000},
				width=0.52\hsize, 
				height=0.35\hsize, 
				xlabel={\texttt{solverALL} [s]},
				cycle multiindex* list = {only marks \nextlist mark=x}, 
			] 
			\compareGroupPlotTop{w/o reduction rules}{lb_disjoint_no_rrules}

			\compareGroupPlot{w/o permanent}{no_permanent_some_rrs}

			\compareGroupPlot{w/o lower bound}{default_no_lb}
			
			\compareGroupPlot{\texttt{CPLEX}}{2cvd_cplex}
		\end{groupplot}
	\end{tikzpicture} 
	\caption{Running time comparison (in seconds) of different configurations of our solver and CPLEX on two datasets (left: Bio33, right: Bio55).
		Each dot represents one instance with the~$x$ and~$y$ coordinates indicating the running time of the respective solver.
		Hence, a dot above (below) the solid diagonal indicates the solver on the $x$-axis ($y$-axis) is faster on the corresponding instance.
		The diagonal lines mark running time factors of~$1$ (solid), $5$ (dashed) and~$25$ (dotted). 
		Dots on the solid horizontal and vertical red lines (at 3600 seconds) indicate timeouts.
		In each plot the running time of our \texttt{solverALL} is displayed at the $x$-axis.
		The $y$-axis shows in each row of the plots a different solver; these are from top to bottom:
		Three configurations of our solver where certain features are disabled (first all data reduction rules, then permanent vertices with the corresponding data reduction rules that require permanent vertices (\cref{rr:2cvd_shrink_permanent,rr:permanent2club,rr:permanent2club_v2,rr:2cvd_prevent_inoptimality}), finally without \cref{lb:w_disjoint_bound}).
		The last row shows the comparison against CPLEX.
	}
	\label{fig:comparisons}
\end{figure}
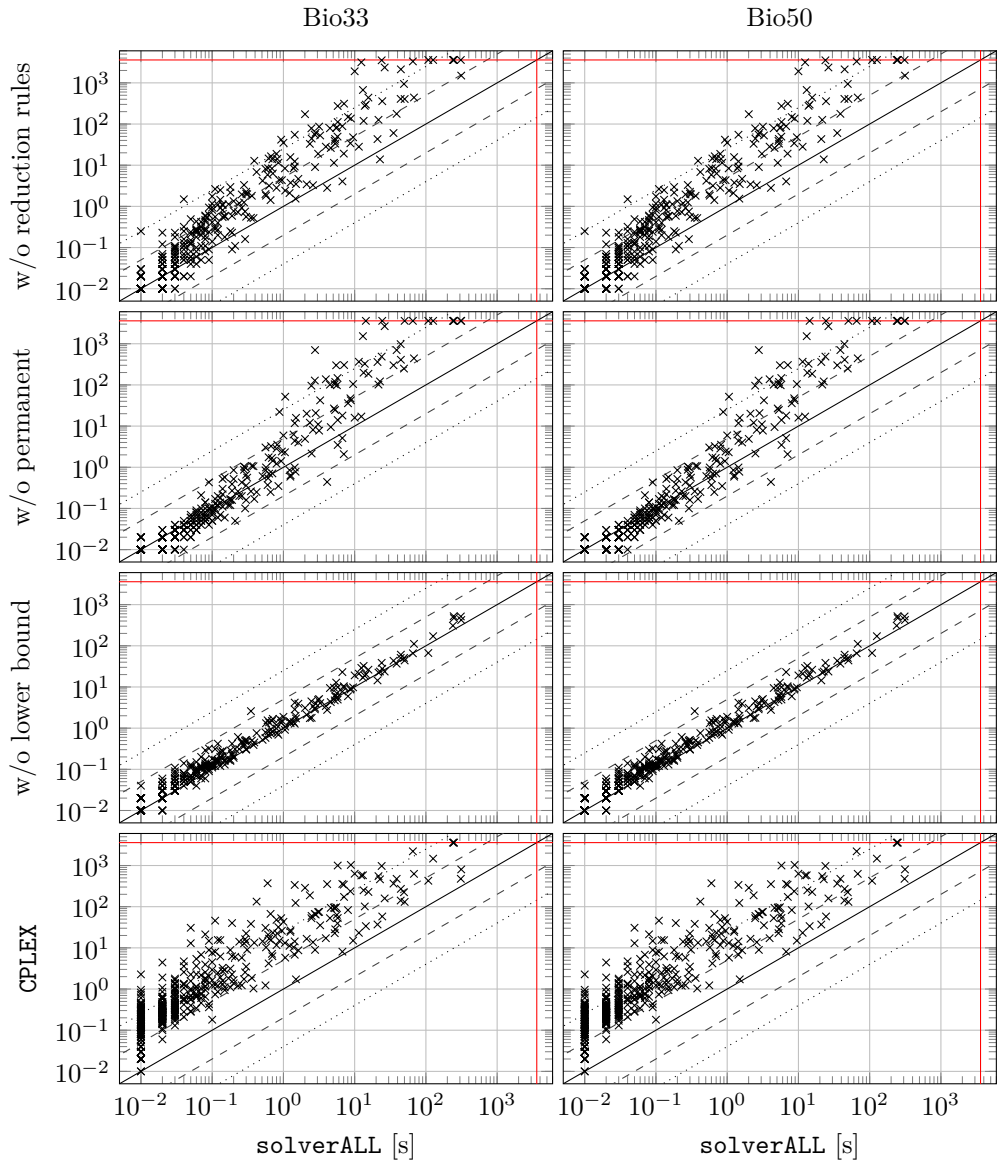
The first row of plots in \cref{fig:comparisons} shows that if we deactivate the data reduction rules, then the performance becomes much worse, especially on the harder instances that require more than 10 seconds to solve.
On average, \texttt{solverALL} (with all data reduction rules) is 6.7 times faster on the Bio33 instances and 3 times faster on the Bio50 instances.
This is in stark contrast to the kernelization lower bound given in \cref{thm:nopolykernel} and gives hope to find small parameters based on which one may perform a mathematical analysis yielding polynomial kernel sizes.

The plots in the second row of \cref{fig:comparisons} show the effect of turning off permanent vertices and the corresponding data reduction rules (\cref{rr:2cvd_shrink_permanent,rr:permanent2club,rr:permanent2club_v2,rr:2cvd_prevent_inoptimality}).
Note that in the Bio50 dataset the variant without permanent vertices is faster on most instances, very likely due to \cref{rr:permanent2club_v2} being an expensive data reduction rule. 
However, the results for the Bio33 dataset show a different picture.
In fact, one can see in both data sets that turning off the feature of permanent vertices solves the easier instances even faster and slows down the solver on the harder instances.
The lack of ``hard'' instances in the Bio50 dataset (see also \cref{fig:impact-n-k}) is the reason for the variant without permanent vertices being faster there.
On average, \texttt{solverALL} (with permanent vertices) is 5 times faster on the Bio33 instances but 1.6 times slower on the Bio50 instances.

The plots in the third row of \cref{fig:comparisons} show that the impact of \cref{lb:w_disjoint_bound} is much smaller than the impact of the data reduction rules and the permanent vertices.
While the \texttt{solverALL} is faster on most instances, the gap does not improve for the harder instances as in the previous two comparisons (see row one and two in \cref{fig:comparisons}).
On average, \texttt{solverALL} (with lowerbounds) is 1.4 times faster on the Bio33 instances and 1.5 times faster on the Bio50 instances.
The plots in the last row of \cref{fig:comparisons} show that our solver is almost always faster than CPLEX by a factor of 5--25 for Bio33 and a factor of 25--100 for Bio50. 
On average, \texttt{solverALL} is 29.3 times faster on the Bio33 instances and 103.6 times faster on the Bio50 instances.
For Bio33 it appears that for harder instances CPLEX is not much slower than our solver.
On Bio50, CPLEX does very poorly compared to our solver. 
This is likely due to the minimum 2-club vertex deletion set size (the parameter in our FPT-algorithm) on these graphs being smaller than for Bio33. %
Moreover, the process for building the ILP model for CPLEX is usually fairly fast, but for larger instances it can take up to 60 seconds. 
For example, in Bio50 there is a graph with 205 vertices and 10455 edges which is already a 2-club cluster graph. 
It takes about 50 seconds to create the ILP model, and when exported to a file it takes up to 1.6GB (uncompressed) and includes 5.8 million constraints, whereas the original graph only takes up 72kB stored in an edge list format.

Summarizing, our solver outperforms a standard ILP-formulation solved with CPLEX.
Moreover, good data reduction rules are crucial to the practical performance of our solver.

\subparagraph*{2-Club Cluster Vertex Deletion Solutions.}
In \cref{fig:cvd_vs_2cvd} (top row) we compare the sizes of a minimum cluster vertex deletion set (CVD) and a minimum 2-club vertex deletion set (2CVD) on our datasets. 
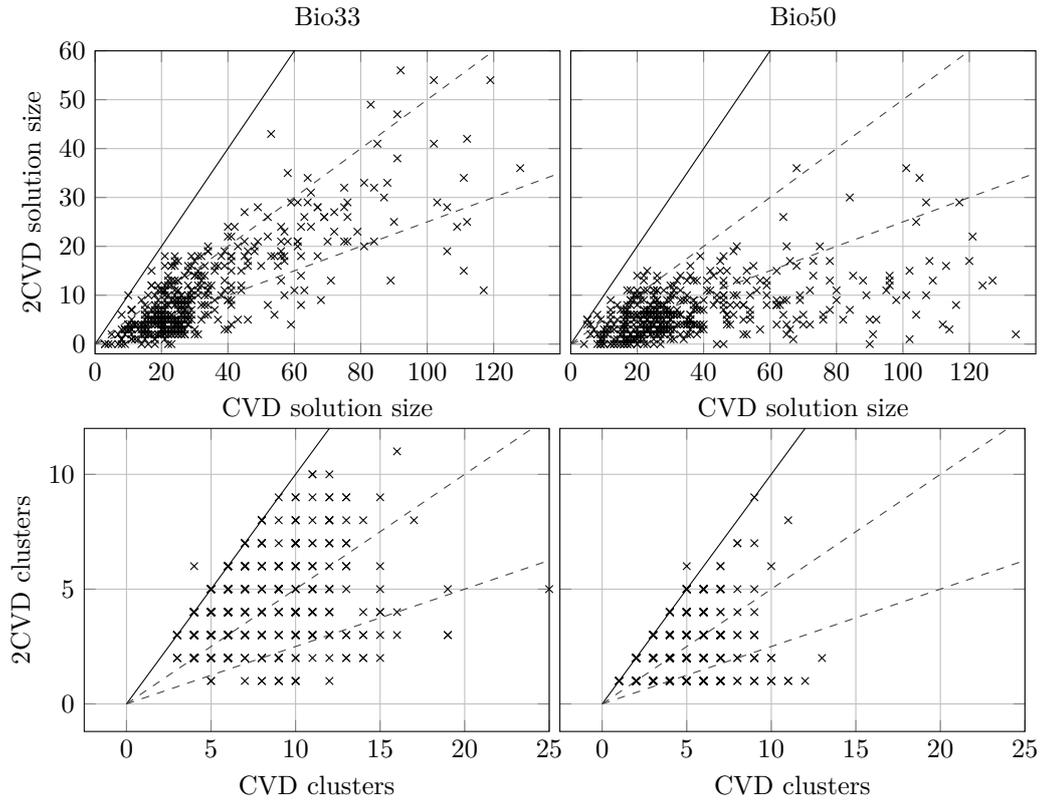
\begin{figure}[t]
	\centering
	\def\maxValueRun{200}
	\def\minValueRun{0}
	\begin{tikzpicture}
		\begin{groupplot}[ 
				group style={ 
					group name=my plots, 
					group size=2 by 2, 
					ylabels at=edge left, 
					yticklabels at=edge left, 
					horizontal sep=4pt
				}, 
				grid,
				xmin=0,
				xmax=140,
				ymin=-2,
				ymax=60,
				ytick={0,10,20,30,40,50,60},
				xtick={0,20,40,60,80,100,120},
				width=0.55\hsize, 
				height=0.4\hsize, 
				xlabel={CVD solution size},
				ylabel={2CVD solution size},
				cycle multiindex* list = {only marks \nextlist mark=x}, 
			] 
		\nextgroupplot[title={Bio33}]
			\addplot+[discard if={cvd__k}{0}] table[col sep=comma,y={2cvd__k}, x={cvd__k}] {clusters33.csv};%
			\addplot[color=black,domain=\minValueRun:\maxValueRun,samples=4] {x};
			\addplot[dashed,color=black!75,domain=\minValueRun:\maxValueRun,samples=4] {0.5*x};
			\addplot[dashed,color=black!75,domain=\minValueRun:\maxValueRun,samples=4] {0.25*x};

		\nextgroupplot[title={Bio50}]
			\addplot+[discard if={cvd__k}{0}] table[col sep=comma,y={2cvd__k}, x={cvd__k}] {clusters50.csv};%
			\addplot[color=black,domain=\minValueRun:\maxValueRun,samples=4] {x};
			\addplot[dashed,color=black!75,domain=\minValueRun:\maxValueRun,samples=4] {0.5*x};
			\addplot[dashed,color=black!75,domain=\minValueRun:\maxValueRun,samples=4] {0.25*x};
		\end{groupplot}
	\end{tikzpicture} 
	
	\begin{tikzpicture}
		\begin{groupplot}[ 
				group style={ 
					group name=my plots, 
					group size=2 by 2, 
					ylabels at=edge left, 
					yticklabels at=edge left, 
					horizontal sep=4pt
				}, 
				grid,
				xmax=25,
				ymax=12,
				width=0.55\hsize, 
				height=0.4\hsize, 
				xlabel={CVD clusters},
				ylabel={2CVD clusters},
				cycle multiindex* list = {only marks \nextlist mark=x}, 
			] 
		\nextgroupplot
			\addplot+[discard if={cvd__k}{0}] table[col sep=comma,y={2cvd__clusters}, x={cvd__clusters}] {clusters33.csv};%
			\addplot[color=black,domain=\minValueRun:\maxValueRun,samples=4] {x};
			\addplot[dashed,color=black!75,domain=\minValueRun:\maxValueRun,samples=4] {0.5*x};
			\addplot[dashed,color=black!75,domain=\minValueRun:\maxValueRun,samples=4] {0.25*x};

		\nextgroupplot
			\addplot+[discard if={cvd__k}{0}] table[col sep=comma,y={2cvd__clusters}, x={cvd__clusters}] {clusters50.csv};%
			\addplot[color=black,domain=\minValueRun:\maxValueRun,samples=4] {x};
			\addplot[dashed,color=black!75,domain=\minValueRun:\maxValueRun,samples=4] {0.5*x};
			\addplot[dashed,color=black!75,domain=\minValueRun:\maxValueRun,samples=4] {0.25*x};
		\end{groupplot}
	\end{tikzpicture} 
	\caption{Comparison of the solution size (top row) and average cluster size for \clustervertexdelete{} ($\hat{=}{}$\textsc{1-Club Cluster Vertex Deletion}) and \twoclubvertexdelete{} for two datasets (left: Bio33, right: Bio55). The solid line is~$y=x$, and the dashed ones are~$y=0.5 x$ and~$y=0.25 x$}
	\label{fig:cvd_vs_2cvd}
\end{figure}
For Bio33 there is a much stronger correlation between these two values than for Bio50. 
For Bio33 the CVD solution size is around 2--4 times larger than the 2CVD solution size, but on many Bio50 instances the CVD solution size can be very large while the 2CVD solution size is below five.

We next compare the number of clusters created by both problems, as shown in the bottom row of \cref{fig:cvd_vs_2cvd}.
As expected, the 2CVD solution creates on most instances much less clusters (while deleting fewer vertices).
Note that all solvers we employ compute minimum-size solutions where the number of clusters is not optimized.
Thus, if there are multiple optimal solution that create a different number of clusters, then we have no control which optimal solution is picked.
We believe that this issue causes two instances (one Bio33 and one Bio50 instance) having a smaller number of clusters when using the CVD solution (see the two points above the solid line in \cref{fig:cvd_vs_2cvd}).

Summarizing, using \twoclubvertexdelete{} rather than \clustervertexdelete{} results in fewer deletions and fewer clusters (which are thus of larger size).

\section{Conclusion}\label{sec:conclusion}

We investigated the problem of modifying graphs into 2-club cluster graphs. 
We have shown that \twoclubedit{} is W[2]-hard for the parameter solution size~$k$. %
Furthermore, we developed and engineered a competitive branch\&bound algorithm algorithm for the fixed-parameter tractable \twoclubvertexdelete{} problem. 
On the theoretical side, we left open, however, whether our ``no-poly-kernel'' result for \twoclubvertexdelete{} parameterized by solution size transfers to the closely related \textsc{2-Club Cluster Edge Deletion} problem, a further open problem from the 
literature~\cite{CDFG20,LZZ12}.
Moreover, it would be interesting to see whether our results also generalize to using~$s$-clubs with~$s \geq 3$.
For other 2-club related graph modification problems to be studied one could consider overlapping clusters~\cite{FGKNU11} or use stricter 2-club models such as well-connected 2-clubs~\cite{Kom+19}. 
Limiting the number of local manipulations~\cite{KU12} is another restriction worthwhile investigations.
On the empirical and algorithm engineering side, note that %
while our solver showed strong performance when working with biological data,
preliminary experiments showed that this is less so when attacking social network data. The reasons for this remained open.

\bibliographystyle{abbrvnat-nodoi}
\bibliography{extracted}

\end{document}